%% file: main.tex
\documentclass[acmtog]{acmart}

\AtBeginDocument{%
  }

\setcopyright{none}
\acmYear{}
\copyrightyear{}
\acmVolume{}
\acmNumber{}
\acmArticle{}
\acmMonth{}
\acmDOI{}
\acmISBN{}

\usepackage{xcolor}
\usepackage{caption}
\usepackage{booktabs}
\usepackage{multirow}
\usepackage{bm}
\usepackage{enumitem}
\usepackage{dblfloatfix}
\input{macros}

\makeatletter
\renewcommand{\@formatdoi}[1]{}
\def\@acmBadgeR{}\def\@acmBadgeL{}
\renewcommand\footnotetextcopyrightpermission[1]{}
\makeatother
\begin{document}

\title{Intrinsic and Triangulation-Agnostic Attention: A  Simple and Powerful Approach for Learning on Meshes}

\author{Ashwath Shetty}
\affiliation{%
  \institution{ Université de Montréal \& Mila}
  \city{City}
  \country{Canada}}
\email{ashwath.shetty@umontreal.ca}
\author{Zihan Zhu}
\affiliation{%
  \institution{Université de Montréal \& Mila}
  \country{Canada}}
\email{zihan.zhu@umontreal.ca}
\author{Soren Pirk}
\affiliation{%
  \institution{Kiel University}
  \country{Germany}}
\email{soeren.pirk@gmail.com}

\author{Noam Aigerman}
\affiliation{%
  \institution{Université de Montréal \& Mila}
  \country{Canada}}
\email{noam.aigerman@umontreal.ca}

\input{sections/0_abstract}



\include{figures/teaser}

\maketitle
\thispagestyle{empty}
\pagestyle{plain}
\input{sections/1_intro.tex}
\input{sections/2_RW}
\input{sections/3_method}
\input{sections/5_results}
\input{sections/conclusion}

\bibliographystyle{ACM-Reference-Format}
\bibliography{references}
\clearpage

\input{sections/7_figures}

\clearpage
\appendix
\input{sections/supp}


\end{document}

%% file: macros.tex
\definecolor{seabornblue}{HTML}{EAEAF2}
\colorlet{seabornblue25}{seabornblue!25}

\newcommand{\manifold}{\Omega}
\newcommand{\mesh}{\mathcal{M}}

\newcommand{\mass}{\textbf{M}}

\newcommand{\real}{\mathbb{R}}

\newcommand{\tribrac}[1]{\left<#1\right>}

\newcommand{\parr}[1]{\left(#1\right)}
\newcommand{\set}[1]{\left\{#1\right\}}
\newcommand{\bigo}[1]{\mathcal{O}\parr{#1}}
\newcommand{\norm}[1]{\left\|#1\right\|}
\newcommand{\mypar}[1]{%
\par
  \vspace{1ex}          
  \noindent\textbf{#1}  
}

\newcommand{\dfeature}{\textbf{f}}
\newcommand{\dquery}{\textbf{q}}
\newcommand{\dkey}{\textbf{k}}
\newcommand{\dvalue}{\textbf{v}}
\newcommand{\dweight}{\textbf{w}}
\newcommand{\datt}{\mathbf{\alpha}}
\newcommand{\spoint}{x}
\newcommand{\tpoint}{y}
\newcommand{\sset}{\mathcal{X}}
\newcommand{\tset}{\mathcal{Y}}

\newcommand{\cfeature}{{f}}
\newcommand{\cquery}{{q}}
\newcommand{\ckey}{{k}}
\newcommand{\cvalue}{{v}}
\newcommand{\cweight}{{w}}
\newcommand{\catt}{{\alpha}}

%% file: sections/0_abstract.tex
\begin{abstract}
  This work proposes an adaptation of the attention mechanism for triangle meshes. The core observation is that endowing the attention mechanism with critical properties for learning over meshes -- intrinsicality and triangulation-agnosticism -- enables it to attain state-of-the-art results over several learning-based tasks in geometry-processing. The above is achieved by modifying the attention mechanism from the bottom up based on simple principles from geometry-processing. Namely, the quantities used within attention --  queries, keys and values -- are created by an intrinsic, triangulation-agnostic network, and treated as discretizations of continuous functions. From that, we devise an appropriate attention mechanism that operates over triangle meshes through standard FEM discretization of the resulting integrals of the above functions. Surprisingly, as far as we know, this straightforward approach has not been utilized for learning over meshes. Experiments show our method exceeds current state of the art, including both mesh-based architectures as well as point cloud transformers. Namely, we show significant improvements on several common benchmarks and tasks -- predicting canonical high-frequency signals; predicting deformations;  computing dense correspondences, both between full shapes and partial ones; and predicting feature descriptors. 
\end{abstract}

%% file: figures/teaser.tex
\begin{teaserfigure}
\centering
  \includegraphics[width=\textwidth]{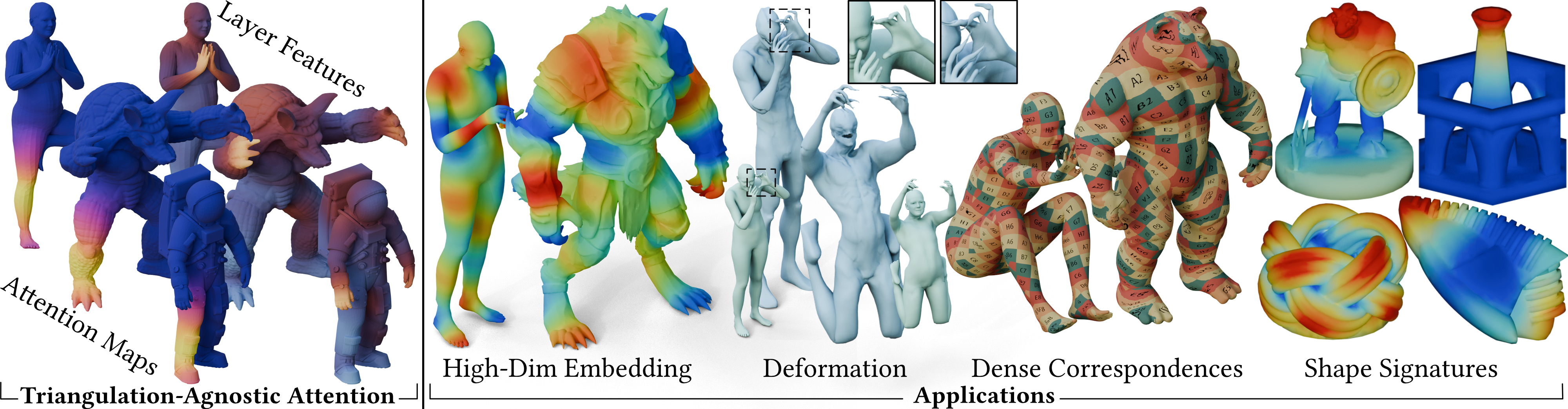}
  \caption{We propose an attention layer designed specifically for meshes, so that it is intrinsic and discretization agnostic. Using these layers we attain state-of-the-art results in several applications: predicting Fourier-like high dimensional embeddings of vertices; deforming characters with a level of granularity unattained before (note fingers); computing mesh-to-mesh dense correspondences ; and, predicting shape signatures. }
  \label{fig:teaser}
\end{teaserfigure}

%% file: sections/1_intro.tex
\section{Introduction}

The attention mechanism~\cite{attention_original}  has become one of the most popular techniques for practical machine learning, as it provides a powerful learning-based method for propagating and aggregating features between different elements; this in turn has been shown to consistently lead to a significant gain in performance compared to, e.g., multi-layer perceptrons (MLPs) or convolutional neural networks (CNNs).  As such, the attention-based transformer architecture has arguably become the default backbone of state-of-the-art machine learning pipelines for text~\cite{gpt}, audio~\cite{verma2021audio}, images~\cite{vit}, and voxels~\cite{VoxFormer}. 

As of now, attempts to apply transformers for learning over \emph{triangle meshes} have not led to the same significant gain in performance that has been observed in other domains. Indeed, state-of-the-art transformer-based architectures for 3D learning usually opt to operate over \emph{point clouds}, discarding the mesh structure (and thus metric and topological information), which we observe leads to under-performance.

On the other hand,  current state-of-the-art  \textit{mesh-based} learning architectures do not utilize an attention mechanism, but rather other, more geometric approaches, e.g.,  a simple multi-layer perceptron coupled with operations based on linear finite elements (FEM), solving a partial differential equation for signal propagation~\cite{diffusionnet,poissonet}. The reason for the success of these geometric networks, in comparison to existing attention-based methods, lies in that they are \textit{intrinsic} (i.e., account for the Riemannian manifold structure of the mesh, and its induced metric) and \textit{triangulation agnostic}, meaning they appropriately treat the signals computed over the mesh as representing a discretization of a \emph{continuous} signal over the continuous manifold the mesh approximates. This, in turn, enables the architecture to leverage the metric and topological information of the mesh, learn and operate over datasets that have a different triangulation for each mesh in the dataset, and, as importantly, perform more efficient learning, similar to the gain of equivariant networks.

Our core contribution stems from the observation that, somewhat unsurprisingly, the above properties of intrinsicality and triangulation-agnosticism are exactly the qualities lacking from the standard attention mechanism, in order for it to provide the desired boost in performance for mesh-based learning. 

Concretely, we revise the attention mechanism from the bottom up, so that it operates in an intrinsic, triangulation-agnostic manner: we treat the usual quantities used in attention (\textit{queries}, \textit{keys}, \textit{values}) as representing discretizations of continuous signals, and compute these via an appropriately intrinsic triangulation-agnostic, neural mesh-based architecture. We then devise appropriate triangulation-agnostic attention operations: for \textit{self} attention between the mesh vertices and themselves, we follow the straightforward analogue of the attention operation for a continuous function, for which we choose a standard FEM-based discretization of a mass-weighted quadrature; this specific choice in turn enables the utilization of existing attention variants~\cite{xformer} that are significantly more computationally efficient and enable us to operate on the same datasets and meshes as previous works. On the other hand, for \textit{cross} attention between mesh vertices and a non mesh-based signal, we show that if the query signal is produced in an  intrinsic triangulation-agnostic manner, then the result from the standard cross-attention mechanism is intrinsic triangulation-agnostic by construction. We provide a theoretical proof that these two discretizations well-approximate the continuous attention definition, with a linearly-bounded error. Figure~\ref{fig:teaser}, left shows attention maps and features computed from our attention mechanism, exhibiting similar structure across several meshes with significantly different geometries and triangulations.

 Surprisingly, in spite of the simplicity of the above approach, we are unaware of any previous work that has proposed it. In fact, as part of our experiments we show that an implementation of this technique could have been achieved using only the existing state of the art as it were \emph{four years} ago (namely, using DiffusionNet~\cite{diffusionnet} as a backbone), and would have still led to a boost in performance that exceeds \emph{current} state of the art (see Figure~\ref{fig:5_single_source_qual}). 

Empirically, we show that our proposed attention layers provide a significantly more powerful 3D learning framework, in comparison to several baselines:  on one hand, we show that our method exceeds the performance of state-of-the-art \textit{point cloud} transformers, due to the attention mechanism being intrinsically aware; on the other hand, we show a significant improvement compared to state-of-the-art \textit{mesh-based} architectures, due to the triangulation-agnosticism of the attention layers. Specifically, we show better results on several applications and benchmark tasks (exhibited in Figure~\ref{fig:teaser}):
\begin{enumerate}
[leftmargin=0.8cm]
    \item \textbf{Predicting intrinsic high-frequency features.} We show a significant improvement in accuracy (6dB in PSNR) compared to the state of the art when predicting Fourier-like features.  
    \item \textbf{Deforming 3D models.} We show a categorical improvement in accuracy and granularity in the prediction of deformations of humans, e.g., achieving articulated finger control (see zoom-in in Figure~\ref{fig:teaser}).
    \item \textbf{Predicting dense correspondences between 3D shapes.} We show that a rather naïve application of our attention mechanism outperforms the current state of the art, and provides highly accurate dense correspondences both for full, as well as partial matching.
    \item \textbf{Predicting surface descriptors.}  We show our method can more accurately predict surface descriptors, e.g., the Heat Kernel Signature~\cite{heatKernelSignature}, than current state of the art. 
\end{enumerate}

%% file: sections/2_RW.tex
\section{Related Work}
\mypar{Attention} The attention mechanism~\cite{attention_original} has seen tremendous success in structured modalities such as vision~\cite{vit}, language~\cite{gpt}, and audio~\cite{verma2021audio}. 
Self attention (i.e., attention between the input tokens and themselves), is often used in learning representations~\cite{dino,dinov3,clip}, while cross attention (i.e., between input tokens and external tokens) serves as an effective way of combining information across modalities~\cite{LDM,multimodal2}. To improve efficiency recent approaches~\cite{flashattention,xformer} avoid the quadratic memory cost associated with standard softmax-weighted attention, by computing the softmax online via highly optimized CUDA kernels. Others propose softmax alternatives  such as linearized attention~\cite{LinearAttention} and state-space models~\cite{mamba}; however, most current methods still rely on softmax weighting~\cite{Qwen,jamba}. We design our mesh-based softmax attention specifically so it can leverage these advances. 

\mypar{Attention in 3D.}  Many works have applied attention for 3D learning problems, e.g., by treating 3D points sequentially~\cite{PTV3,PTV2,PT,litept}, which discards the intrinsic structure. Structured voxels~\cite{meshformer, VoxFormer}, are proven to be a powerful technique, however limited in its ability to operate over intricate geometry which may occupy a small part of space. 
Previous attempts to apply attention on \textit{meshes} either  assume a fixed mesh triangulation~\cite{metro1,METRO2}, devise mesh-based tokenization strategies~\cite{meshmae,meshmamba}  or inject mesh-based encoding by utilizing the eigenvectors of the Laplacian, either as positional encoding~\cite{meshgraphtransformer,laplacianmeshtransformer}, or as tokenization~\cite{recipe}. All these methods are not triangulation agnostic and hence cannot operate in the settings and applications shown in this paper.

The method most similar to ours is HodgeFormer~\cite{hodgeformer}; it proposes to use discrete exterior calculus (DEC) for representing mesh-based attention by learning the intrinsic mesh-based DEC operators themselves end-to-end. These DEC operators are learned solely based on the mesh's discrete connectivity and hence are not triangulation agnostic.  We compare to HodgeFormer, and show our method provides more accurate predictions in Section~\ref{sec:pose-encoding}. Namely, we use the standard discretization of intrinsic differential operators, which provides our framework with  triangulation-agnosticism. 

\mypar{Continuous attention and neural operators.} Many works treat attention through the lens of continuous operators and operator learning~\cite{neuraloperator,galerkin,Transolve}. Galerkin Transformer~\cite{galerkin} draws an explicit connection between softmax-free linear attention and Galerkin-type projections. This has motivated several attention-based neural operators~\cite{neuraloperator,Transolve}, that try to learn the mapping between functions end-to-end (for, e.g., solutions to differential equations).
Our goal is different from these approaches, which try to learn functional operators -- we use the standard, fixed discretization of the mesh-based differential operators, and define a simple, triangulation-agnostic attention mechanism through it. 
\input{figures/disc-agnostic}
\mypar{Learning on meshes.}  Many methods for learning over meshes have been proposed, such as extending convolutions to surfaces~\cite{boscaini2016learning, fey2018splinecnn, monti2017geometric,bronstein2017geometric, dynamic_simonovsky_2017, geodesicConvo_masci_2015,meshcnn_hanocka_2019,deltaconv}, leveraging equivariance~\cite{de2020gauge, he2020curvanet, fieldconvo_mitchel_2021, mdgcnn_poulenard_2018, HSN_wiersma_2020, cgconv_Yang_2021, sun2020zernet}, or treating the mesh as a graph ~\cite{dynamic_simonovsky_2017,meshwalker_lahav_2020}. Others use a Fourier-like decomposition using the eigenvectors of the Laplacian~\cite{Litany_2017_ICCV,Yi_2017_CVPR, halimi2019unsupervised, roufosse2019unsupervised, donati2020deep, attaiki2021dpfm,hodgenet_smirnov_2021}; this, however limits their granularity to that of the truncated eigenbasis, and prevents scaling to larger datasets. For the backbone of our architecture, we utilize two methods considered state-of-the-art for mesh-based learning: PoissonNet~\cite{poissonet} and DiffusionNet~\cite{diffusionnet}. They propagate features over the domain via FEM solutions to PDEs, thus are agnostic to discretizations. On their own, their feature propagation is limited to follow the PDE's solution, and hence much weaker than our attention-based propagation. Furthermore, their conditioning mechanism is limited to standard concatenation of the conditioning signal as an additional input to the backbone MLP; we show (e.g., in the deformation experiment, Section~\ref{sec:shape_deform}) that this is a much weaker conditioning method than applying cross attention between the mesh and the conditioning signal.

\mypar{Learning deformations.}
Computation of deformations is a key task in computer graphics and geometry processing~\cite{sumner2004deformation,bogo2014faust, gao2019sdm,Lipman04,Sorkine04,sumner2004deformation,Yu04}, often achieved via rigs~\cite{skinningcourse:2014,Fulton:LSD:2018,jacobson2011bounded,kavan2008geometric,lipman2008green,ju2005mean}, which several methods leverage using neural networks~\cite{AnimSkelVolNet,RigNet,Holden:inverse_rig:2015,liu2025riganything,wang2019neural,sun2024tuttenet,li2021learning}, often combining them with prediction of nonlinear residuals~\cite{Bailey:2018:FDD,bailey2020fast,Romero:2021,Zheng:secondary_motion:2021,yin2021_3DStyleNet}. Alternatively if the source is a fixed template~\cite{bogo2014faust,varol17_surreal,SMAL:2017,STAR:2020}, one can regress the vertex coordinates directly~\cite{anguelov2005scape,Bogo:ECCV:2016,Shen:2021}. Methods often focus on preserving local details by operating in the gradient space~\cite{NJF,temporalNJF,Gao23,Kim25,Yoo24}.  The current state-of-the-art approach for predicting deformations, PoissonNet~\cite{poissonet}, operates in the gradient domain in each of its layers, leading to highly-accurate predictions. We compare to it and show our attention-based framework attains significantly more accurate and granular predictions of deformations.

\mypar{Shape correspondence.} Computing dense correspondences between 3D shapes is a fundamental task in geometry processing~\cite{RL01,ZSN03,MAKM21}, see \cite{reviewnonrigid} for a more detailed overview. Commonly tackled via handcrafted features and optimization~\cite{heatKernelSignature,sadh04,apl14,SBCK19,3dcoded},  more recent approaches shifted to learned priors~\cite{AAWEOW,MAKM21,neural_semantic_surface_maps}. A common approach is to use  functional maps~\cite{ovsjanikov2012functional} which has many followups using  energy formulations~\cite{ren2018continuous,panine2022non,smoothshells}, utilizing pointwise maps~\cite{zoomout}, or using learned features~\cite{attaiki2021dpfm,Litany2017DeepFM,simplified_fmaps}. Most recently, Diffumatch~\cite{diffumatch} learned a diffusion prior in the spectral domain to attain state-of-the-art results.  We compare to them in Section~\ref{sec:correspondance} and attain higher accuracy via our attention layers.

%% file: figures/disc-agnostic.tex
\begin{figure}[t]
\centering
    \includegraphics[width=1\linewidth]{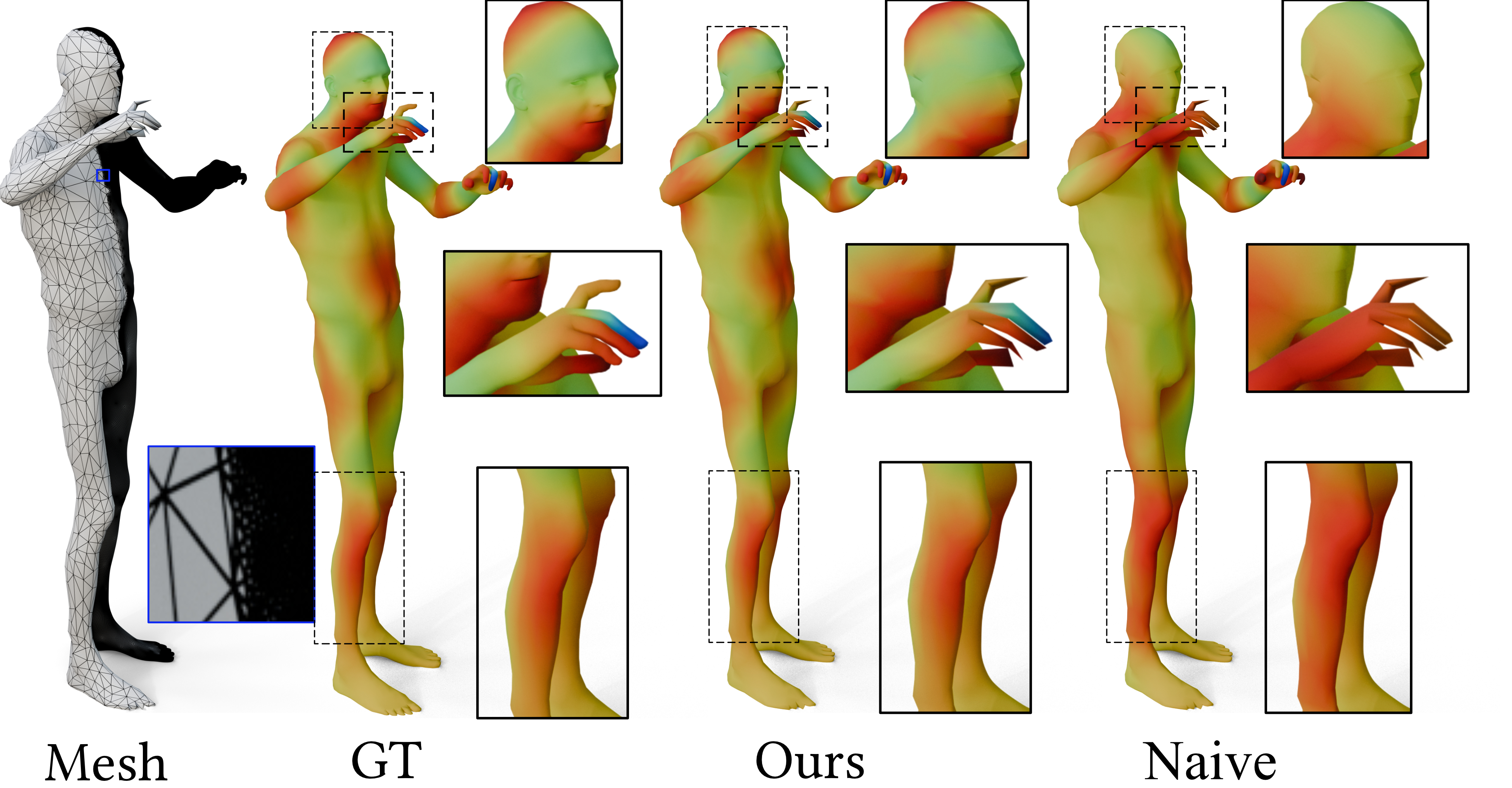}
    \caption{
   \textbf{Triangulation-agnostic vs. naive attention.} When faced with a mesh with different triangulation than that of the training set, \emph{our} triangulation-agnostic attention still predicts signals that match the ground truth; \textit{naive} attention predicts incorrect signals.
    %
 }
    %
    
    %
    \label{fig:3_disc_agnotism_sa}
\end{figure}

%% file: sections/3_method.tex

\section{Method}
\subsection{Preliminaries}
\label{sec:prelim}
We assume to be  given a mesh $\mesh$ with vertices $\sset=\set{\spoint_i}_{i=1}^n$. We further assume the mesh discretizes a continuous surface, $\mesh\sim\manifold$.
\mypar{Triangulation agnosticism.}
 We consider functions defined over the mesh, discretized by assigning a value $\textbf{F}_i$ to each vertex $\spoint_i$ of the mesh. The discrete function $\textbf{F}$ is said to be  \textit{triangulation agnostic } if it is generated in a way that entails it approximates an (unknown) continuous function $F(\spoint)$ defined over the manifold, 
 \begin{equation}
 \label{eq:tri_agnostic}
     \mathbf{F}_i\approx F\parr{\spoint_i}. 
 \end{equation}
 If the above equation holds, then computing this function on two different triangulations $\mesh^A,\mesh^B$ of the surface would yield two functions $\textbf{F}^A,\textbf{F}^B$ which are near-identical, $\textbf{F}^A\approx F \approx\textbf{F}^B$.

\mypar{The attention mechanism.} 
We next describe the standard attention mechanism as it is commonly used. For coherence with our mesh-based construction, we describe the mechanism as operating over \textit{point sets}. Specifically, attention operates on a given point $\spoint$ and another set of points $\tset=\set{\tpoint_1,...,\tpoint_m}$. Each point in $\tset$ has a corresponding vector of features called \emph{values}, $\dvalue_i\in\real^{|F|}$, $V=\set{\dvalue_1,...,\dvalue_m}$.  

The goal of the attention mechanism is to aggregate these features  defined on $\tset$ into a feature on $\spoint$. Namely, $\spoint$ has a vector associated with it, called the 
\emph{query} $\dquery\in\real^r$. 
Each point $\tpoint_i$ has a similar vector called a 
\emph{key} $\dkey_i\in\real^r$. The 
weighting that $\spoint$ places on the features of $\tpoint_j$ is defined by taking the inner product between the query and the corresponding key and exponentiating it:
\begin{equation}
    \label{eq:dweight}
    \dweight_{j}= e^{\tribrac{\dquery,\dkey_j}}.
\end{equation}

Each unnormalized weight is finally transformed into the \emph{attention weight} $\alpha_{j}$ by a global normalization that requires that all attention weights sum to $1$:

\begin{equation}
    \label{eq:datt}
    \datt_{j}=\frac{\dweight_{j}}{\sum_k \dweight_{k}}.
\end{equation}
This defines the attention weight as a \emph{softmax} function of the inner product of the query with the key.
The final aggregated feature $\dfeature$ of point $\spoint$ is defined by the weighted sum
\begin{equation}
\label{eq:dweighted_sum}
    \dfeature= \sum_j \datt_{j}\dvalue_j=\frac{\sum_j e^{\tribrac{\dquery,\dkey_j}}\dvalue_j}{\sum_k e^{\tribrac{\dquery,\dkey_k}}}.
\end{equation}

\subsection{Problem Statement}
 We now move on to our main goal: defining intrinsic, triangulation-agnostic attention over meshes. 
 
 To apply attention to a mesh, the attention mechanism is applied to each mesh vertex  in parallel. Namely,  each vertex $\spoint_i$ is assigned with a query vector $\dquery_i$;  using it, the attention mechanism outputs the aggregated feature vector  $\dfeature_i$ for that vertex. 
 
 As our goal is to ensure intrinsicality and triangulation agnosticism, we require the features satisfy Equation~\eqref{eq:tri_agnostic} and represent sampled values of a continuous function $\cfeature$ defined over the underlying surface, i.e., $\dfeature_i \approx\cfeature(\spoint_i)$; thereby we will ensure the output yields similar $\dfeature_i$ values for different triangulations.
 

To achieve such an intrinsic, triangulation-agnostic attention output, we need to ensure that in turn every step and quantity used \emph{within} the attention mechanism is computed in an intrinsic, triangulation agnostic manner. Practically, this entails two requirements: \begin{enumerate}
    \item The queries $\dquery_i$ themselves should be intrinsic, triangulation-agnostic and represent a continuous function $\cquery(\spoint)$.
    \label{item:q_agnostic}
    \item The attention mechanism should operate in  a triangulation-agnostic manner over the queries.
    \label{item:a_agnostic}
\end{enumerate}


 It is straightforward to satisfy Item~\eqref{item:q_agnostic} and achieve triangulation-agnosticism for the queries: we note that $\dquery$ can simply be predicted by an appropriate intrinsic, triangulation-agnostic network such as PoissonNet~\cite{poissonet} or DiffusionNet~\cite{diffusionnet}, whose output is triangulation agnostic (as long as the input is as well). Henceforth, we assume the query vector approximates a continuous function $\dquery_i\approx\cquery\parr{\spoint_i}$.

 We next turn to tackle Item~\ref{item:a_agnostic} and provide a triangulation-agnostic attention mechanism. We tackle separately each of the two main attention types: \emph{cross attention} between a mesh and a discrete set of external parameters (Section~\ref{sec:cross_attn}), and \emph{self attention} between the mesh's vertices and themselves (Section~\ref{sec:self_attn}).

\subsection{Attention Between the Mesh and a Discrete Set}
\label{sec:cross_attn}
We first tackle \emph{cross-attention} between the mesh and another, discrete set. Namely, assume both the values $\dvalue_i$ to be aggregated and their corresponding keys $\dkey_i$ are still defined over some other, non mesh-based set $\tset$ as described in Section~\ref{sec:prelim}. In this case, we show triangulation agnosticism is attained using the standard, discrete attention mechanism, Equation~\eqref{eq:dweighted_sum}, as long as the queries  $\dquery_i$ are triangulation agnostic themselves. 

Indeed,  per triangulation-agnosticism of the queries, assume $\dquery$ well-approximates an underlying continuous function $\cquery(\spoint)$, $\dquery_i=\cquery(\spoint_i)+E$, with $\norm{E}=\varepsilon$ being the discretization error. Then, assuming the keys and values are bounded, it holds that the error between the computed feature $\dfeature_i$ and the ``true''  value of the continuous feature function $\cfeature(\spoint_i)$ depends solely on the approximation error between the computed query $\dquery_i$ and the ``true'' query function $\cquery(\spoint_i)$, and is linear in it: $\norm{\dfeature_i-\cfeature\parr{\spoint_i}}=\bigo{\varepsilon}$. We provide a straightforward formal proof in Appendix~\ref{sec:cross_attn_proof} in the supplementary.

\subsection{Attention Between the Mesh and Itself}
\label{sec:self_attn}
We now turn to defining \emph{self-attention} over the mesh, i.e., each vertex aggregates values from other vertices, and the keys $\dkey_i$ and features $\dfeature_i$ are also defined over the vertices $\spoint_i$ of the mesh; similarly to the queries, we compute these quantities via a triangulation-agnostic backbone, thereby treating them as approximating continuous functions, $\dkey\approx\ckey(\spoint), \dvalue\approx\cvalue(\spoint)$. We next define \emph{continuous} attention with respect to these continuous functions. We will then correctly discretize this operation, from an FEM perspective, leading to triangulation-agnostic output features, which we prove well-approximates the continuous values.

\mypar{Continuous, surface-based attention.} 
We reformulate the attention equations from Section~\ref{sec:prelim}, treating all quantities as continuous functions. Given the query and key functions $\cquery\parr{s}, \ckey\parr{t}$, the weight computation between two vertices (Equation~\eqref{eq:dweight}) becomes

\begin{equation}
    \label{eq:cweight}
    \cweight\parr{s,t}= e^{\tribrac{\cquery\parr{s},\ckey\parr{t}}}.
\end{equation}
The summation operations are now replaced with continuous integrals, namely the normalized attention weights (Equation~\eqref{eq:datt}) become
\begin{equation}
\label{eq:catt}
    \catt\parr{s,t}=\frac{\cweight\parr{s,t}}{\int  \cweight\parr{s,r}\text{d}r},
\end{equation}
and the weighted sum (Equation~\eqref{eq:dweighted_sum}) becomes a weighted integral
\begin{equation}
\label{eq:cweighted_sum}
    \cfeature\parr{s} = \int \catt\parr{s,t}\cvalue\parr{t} \text{d}t.
\end{equation}

\mypar{Discretizing the continuous attention.} As before, our final goal is to compute the aggregated features, in this case consisting of the evaluation of the continuous function $\cfeature$ over the vertices, $\cfeature\parr{\spoint_i}$. 

In order to compute this quantity, we need to compute the integrals in Equations~\eqref{eq:catt} and~\eqref{eq:cweighted_sum}. One discretization would be to take $\cquery,\ckey,\cvalue$ as piecewise-linear functions, interpolating the vertex values. In this case, the integrands are an exponent of a linear function, whose integral has a closed-form identity and thus can be computed exactly. However, based on our experiments this leads to expensive computation, and furthermore prevents integration into existing memory-efficient attention kernels \cite{xformer,flashattention}.

Hence, we instead opt to use the standard vertex-based lumped-mass weighting  quadrature approximation of the integrals, which is highly-efficient and furthermore, 
lends itself to standard attention formulation and thus enables us to leverage existing optimized CUDA kernels (see supplemental, Section~\ref{sec:architecture}). 
Specifically, for the denominator in Equation~\eqref{eq:catt} we get
\begin{equation}
    \int  \cweight\parr{\spoint_i,r}\text{d}r \approx\sum_j\mass_{jj} \dweight_{ij},
\end{equation}
where $\mass_{jj}$ is the lumped mass of point $\spoint_j$.
Hence,  the attention weight is 
\begin{equation}
    \catt\parr{\spoint_i,\spoint_j}\approx\frac{\dweight_{ij}}{\sum_k\mass_{kk} \dweight_{ik}},
\end{equation}

Similarly, for the integral in Equation~\eqref{eq:cweighted_sum} we get 
\begin{equation}
    \int \catt\parr{s,t}\cvalue\parr{t} \text{d}t\approx\sum_j \mass_{jj} \catt\parr{\spoint_i,\spoint_j}\cvalue\parr{\spoint_j}.
\end{equation}
Thus, plugging in $\cvalue\parr{\spoint_i}\approx\dvalue_i$ the final aggregated feature is reduced to a simple mass-weighted version of the regular attention mechanism:
\begin{equation}
    \dfeature_i = \frac{\sum_j\mass_{jj}\dweight_{ij}\dvalue_j}{\sum_k\mass_{kk} \dweight_{ik}}.
\end{equation}
In spite of this equation being a division of two quadratures by one another, we prove that it still leads to a good approximation of the continuous function $\cfeature(\spoint)$: assuming $\dquery,\dkey,\dvalue$ are triangulation-agnostic and good approximations (up to $\varepsilon$) of the underlying continuous functions $\cquery,\ckey,\cvalue$, then the approximation error  $E=\dfeature_i-\cfeature(\spoint_i)$ is $\bigo{\varepsilon + h^2}$, with $h$ being the maximal edge length of the mesh. We prove this in Appendix~\ref{sec:self_attn_proof}.   Figure ~\ref{fig:3_disc_agnotism_sa} qualitatively illustrates the importance of mass weighting in triangulation agnosticism: given a heavily re-triangulated mesh, a network with a mass-weighted attention layer is able to recover the GT signal, whereas naive attention fails. (For more details about the setup of this experiment, refer to Section~\ref{sec:disc-agnostic-validation})
\subsection{Incorporating the Attention Layers}
We incorporate our attention layers as follows: inside each block, $\dquery$ and possibly $\dkey,\dvalue$ are predicted by a triangulation-agnostic network~\cite{poissonet} which is input with the output of the previous block. For cross-attention (resp. self-attention)  we utilize $\dquery$ (resp. $\dquery,\dkey,\dvalue$) as described in Section~\ref{sec:cross_attn} (resp. \ref{sec:self_attn})  to output the features $\dfeature$. See supplementary, Section~\ref{sec:architecture} for full details on the incorporation of these layers inside a learning architecture.

%% file: sections/5_results.tex
\section{Experiments}

We now evaluate our method on several benchmarks, showing it outperforms both mesh and point cloud state-of-the-art architectures.  In all experiments, our method uses a similar or fewer number of parameters than that of the baselines. Unless specified otherwise, \textbf{Ours} refers to a network with PoissonNet and our proposed attention layers added to it.  
{Refer to the supplemental for full details of the implementation of our models and the baselines.}
\mypar{Datasets.}
\label{sec:dataset}
In order to evaluate our methods over large datasets with groundtruth correspondence between models, we follow previous works~\cite{NJF,poissonet} and use the SMPL parametric model~\cite{SMPL-X_pavlakos_2019} to generate different humans in various poses, drawn from a combination full-body yoga poses~\cite{moyo} and human-object poses with elaborate finger articulation~\cite{GRAB:2020}.


\mypar{Evaluation Metrics.}
For tasks involving function prediction on surfaces (e.g., positional encoding), we report mean squared error (MSE) and peak signal-to-noise ratio (PSNR). For shape deformation tasks (Sec. \ref{sec:shape_deform}), we measure vertex-wise MSE (using the given ground truth), chamfer distance (CD), and Hausdorff distance (HD).

\input{figures/positional_encoding_qual}
\input{figures/qual_position_encoding_flex}
\subsection{Predicting High Frequency Intrinsic Signals}
\label{sec:pose-encoding}

We begin by evaluating our method's ability to capture high-frequency intrinsic functions over meshes. To construct the training set, we compute the Laplacian eigenvectors of a \textit{single} SMPL mesh, and then transfer these functions to all other meshes in the training set using the 1-to-1 correspondence between SMPL meshes, thereby providing a canonical set of functions. We then train our framework to predict these functions over different meshes. 

\subsubsection{Discretization-Agnosticism of Mass-Weighting}
\label{sec:disc-agnostic-validation}
As Figure ~\ref{fig:6_qual_pos} shows, our method accurately predicts high-frequency functions on meshes it was not trained on (bottom row). Similarly, Figure~\ref{fig:qual_pos_enc} demonstrates generalization to a variety of out-of-distribution (OOD) shapes. Figure~\ref{fig:3_disc_agnotism_sa} isolates the role of the mass weighting: on a challenging discretization (half simplified, half subdivided), a network trained with the mass-weighting layer recovers the ground-truth eigenfunction, whereas naive attention fails. Table~\ref{tab:robustness} quantifies this difference on a subset of re-triangulated meshes, reporting PSNR with respect to the ground truth.

\subsubsection{Comparisons.}
\input{figures/single_source_ca_results}
The above task serves as an excellent benchmark to compare the accuracy of our method to three state-of-the-art baselines: 1) PTV3~\cite{PTV3}, a state-of-the-art 3D \emph{point-clouds} transformer; 2) HodgeFormer~\cite{hodgeformer}, a transformer that operates over meshes via discrete exterior calculus (DEC); 3) PoissonNet~\cite{poissonet}, a state-of-the-art intrinsic, triangulation-agnostic method for learning over meshes. 

Table~\ref{tab:1_pos_encoding}  quantitatively verifies that our method significantly outperforms these baselines, while utilizing the smallest number of parameters. Note we evaluate on the original SMPL triangulation used during training -- this verifies that our intrinsic awareness of the mesh's  Riemannian structure directly improves performance, regardless of any need to operate on different triangulations. We compare to two variants of  PTV3, with $\times5$ and $\times90$ more parameters than ours, both of them lack any awarness of the underlying mesh structure and hence are less effective than us.  Similarly, HodgeFormer attempts to \textit{predict} intrinsic operators but does not utilize the ones stemming from the induced Riemannian metric; Note all of these methods' performance further deteriorate when operating on meshes with different triangulations (Figure~\ref{fig:6_qual_pos}, bottom row). While PoissonNet is triangulation-agnostic, it does not use an attention mechanism, hence it lacks the ability to accurately capture high frequencies, leading to an erroneous result on the fingers in Figure~\ref{fig:6_qual_pos}.

\subsection{Shape Deformation}
\label{sec:shape_deform}

We next tackle \textit{shape deformation}, a fundamental task in geometry processing, which also serves as a common benchmark for evaluating mesh-based learning methods~\cite{NJF,poissonet}. In this task, the network is trained on source-target pairs of SMPL~\cite{SMPL-X_pavlakos_2019} meshes, where the network receives the source mesh, along with the target's SMPL pose parameters. The network is then trained to deform the source mesh, i.e., predict a new position to each of its vertices so that they align with the target vertices, conditioned on the given pose parameters. To condition our framework on the pose parameters, we utilize the intrinsic cross-attention mechanism (Section~\ref{sec:cross_attn}), aggregating values from the entries of the pose parameters vector into the vertices of the mesh in each block. We compare our method to PoissonNet~\cite{poissonet}, which is considered the state of the art in learning of deformations.
\subsubsection{Single source deformation.}
\label{sec:single_source}
We begin by evaluating the ability of our method to predict deformations of a single fixed-source mesh (we use the T-Pose of an SMPL human). Figure~\ref{fig:5_single_source_qual} shows that our method can capture fine-grained finger articulation, due to the attention layers, while vanilla, attention-less PoissonNet does not attain such accurate finger articulation. Table~\ref{tab:single_source} verifies this gap.

\subsubsection{Retroactive application of attention to 2022 state of the art.} While machine learning often relies on cutting-edge recent developments, we emphasize our attention mechanism could have been directly implemented several years ago. We experimentally show this by implementing a variant which incorporates our attention mechanism with only the means available in 2022, namely a DiffusionNet~\cite{diffusionnet} backbone. Indeed,
Figure~\ref{fig:5_single_source_qual} and Table~\ref{tab:single_source} show that using \textit{2022}'s DiffusionNet with our attention layers still outperforms \textit{2026}'s current state of the art -- PoissonNet~\cite{poissonet} (without any attention). This stems from cross attention providing a much more elaborate conditioning mechanism than simple concatenation of additional signals as inputs to the layers. 
\input{figures/random_source_qual}

\subsubsection{Arbitrary source deformation.}
\label{sec:random_source}
\input{tables/quan_tria_agnotism}
\input{tables/quan_position_encoding}
\input{tables/quan_single_source}
As a more challenging task, we evaluate the ability of our method to predict deformations when the source is an arbitrary SMPL human in an arbitrary pose, by sampling random source/target pairs from our dataset (Section~\ref{sec:dataset}) for training and testing. We additionally feed as input to the first block of the pipeline the high-frequency signals discussed in Section~\ref{sec:pose-encoding}, predicted using our self-attention blocks. These serve as an intrinsic \emph{positional encoding} for each vertex.  

 Once our network is trained, triangulation agnosticism enables it to generalize and accurately deform non-SMPL models, as shown in Figure~\ref{fig:application_articulations}. Note that our approach is able to capture accurate finger articulations for arbitrary shapes, which was not possible with previous techniques.

\subsubsection{Ablation and comparison.}Table~\ref{tab:3_random_source} further shows an ablation in which we strip away different parts of the pipeline -- cross attention, self attention, and the positional encoding -- eventually leaving only attention-less PoissonNet to compare with. Each removal reduces accuracy, showing the significant contribution of the attention layers.

As part of this ablation, we observe that removing the positional encoding slightly reduces accuracy, however makes the trained model generalize better to non-humans (e.g., having four fingers instead of 5).  Figure~\ref{fig:qual_comparison_random_source} shows examples of deformations of such non-human models; though producing less granular deformations, this ablated method still exceeds the current state of the art -- we attribute this to the use of the triangulation-agnostic cross attention as a conditioning technique. 

\input{figures/random_source_qual_wo_pe}

\input{tables/quan_random_source}
\input{tables/quan_correspondances}
\input{figures/qual_flex_correspondances}
\input{figures/hks_results}
\input{figures/application_position_encoding}

\subsection{Correspondences}
\label{sec:correspondance}
Due to the granularity of our method's prediction,  we find that it can also serve as an effective model for computing correspondences. In order to show the accuracy of the framework, we employ the simplest possible matching strategy: given a source shape and a target shape, we use the trained model from Section~\ref{sec:pose-encoding} to predict per-vertex features, normalize them to unit norm, and perform Euclidean nearest-neighbor matching. We additionally augment the training set with random rotations and omit random patches from the mesh. 

\subsubsection{Comparison to state-of-the-art correspondence methods.} We compare our approach against a state-of-the-art technique specifically designed for this task -- DiffuMatch~\cite{diffumatch}, along with previous methods that work has compared to. We follow their testing protocol, and report the average geodesic error to the ground truth mapping in Table~\ref{tab:quan_corresponance}, using the reported scores from their Table 1 for the other methods. Surprisingly, without any training for correspondence with source-target pairs, our simple approach is superior to the current state of the art on four out of five benchmarks (we find that the lack of success on the remaining dataset (DT4D-Intra) is due to some non-humanoid outliers, e.g., a cartoon mouse with one finger).

Figure~\ref{fig:qual_corr} visualizes the correspondences by transferring a textured UV map from the source to the target mesh, using the computed correspondences. Our method produces continuous, meaningful correspondences. While DiffuMatch produces reasonable results for more human-like meshes, its accuracy degrades for out-of-distribution meshes (alien,  bottom row).

We further stress-test the accuracy of our predicted features by computing \textit{partial matches} between a partial mesh and a full one. This shows our attention-based feature predictor's power and robustness. In contrast, DiffuMatch~\cite{diffumatch} suffers a significant degradation in results (See Fig \ref{fig:qual_corr} bottom row). See the supplemental for more details on these experiments.

\subsubsection{Comparison to other feature extractors.}  Other methods have designed large point-based feature extractors and showed they achieve state-of-the art results for correspondence. We compare to two of them, as shown in Figure~\ref{fig:6_application_position_encoding}:
Utonia~\cite{utonia} is a large, 135-million-parameter model, trained on a dataset of more than a million samples; it is nonetheless a point-based, non-intrinsic method that lacks the inductive bias to enable it to produce continuous, dense mappings between regions, resulting in a highly discontinuous map. SA3DF~\cite{sa3df} distills vision features generated via multiview renderings onto a mesh, and is trained with a contrastive loss involving surface geodesics; however, beyond the training loss, SA3DF  does not utilize intrinsic information during prediction, and further relies on aggregation of vision features rendered from multiple views. This leads to under-performance compared to our method.




\subsection{Predicting Heat Kernel Signatures}
\input{tables/quan_hks}
To evaluate the ability of our method to predict signals on general meshes, we follow the same experiment and dataset from PoissonNet~\cite{poissonet} (their Section 6.5, "Learning local signals") and train our method to predict Heat Kernel Signature~\cite{heatKernelSignature} features (HKS) using their dataset of ground truth HKS features on Thingi10K~\cite{Thingi10K} meshes. Table~\ref{tab:hks_results_quan} shows that our method is twice as accurate as PoissonNet. Figure~\ref{fig:hks_results} shows that the errors PoissonNet introduces stem from subtle geometric features that are hard to account for without an attention mechanism.

\subsection{Efficiency Analysis}
\input{tables/runtime}
Tab. \ref{tab:runtime} shows a speed/memory analysis of running PoissonNet with and without our attention layers, thereby highlighting the attention's impact on efficiency. We show this analysis for two experiments, in each of which we isolate each of the two attention types.   Cross attention operates between vertices and low-dimensional signals, and hence comes at a marginal extra cost for speed and memory. Self attention leads to a significant slowdown, however, in terms of memory, note that since our choice of FEM discretization (Section~\ref{sec:self_attn}) enables the use of a memory-efficient attention implementation~\cite{xformer}, memory consumption does not increase.  

%% file: figures/positional_encoding_qual.tex
\begin{figure}[t]
\centering
    \includegraphics[width=1\linewidth]{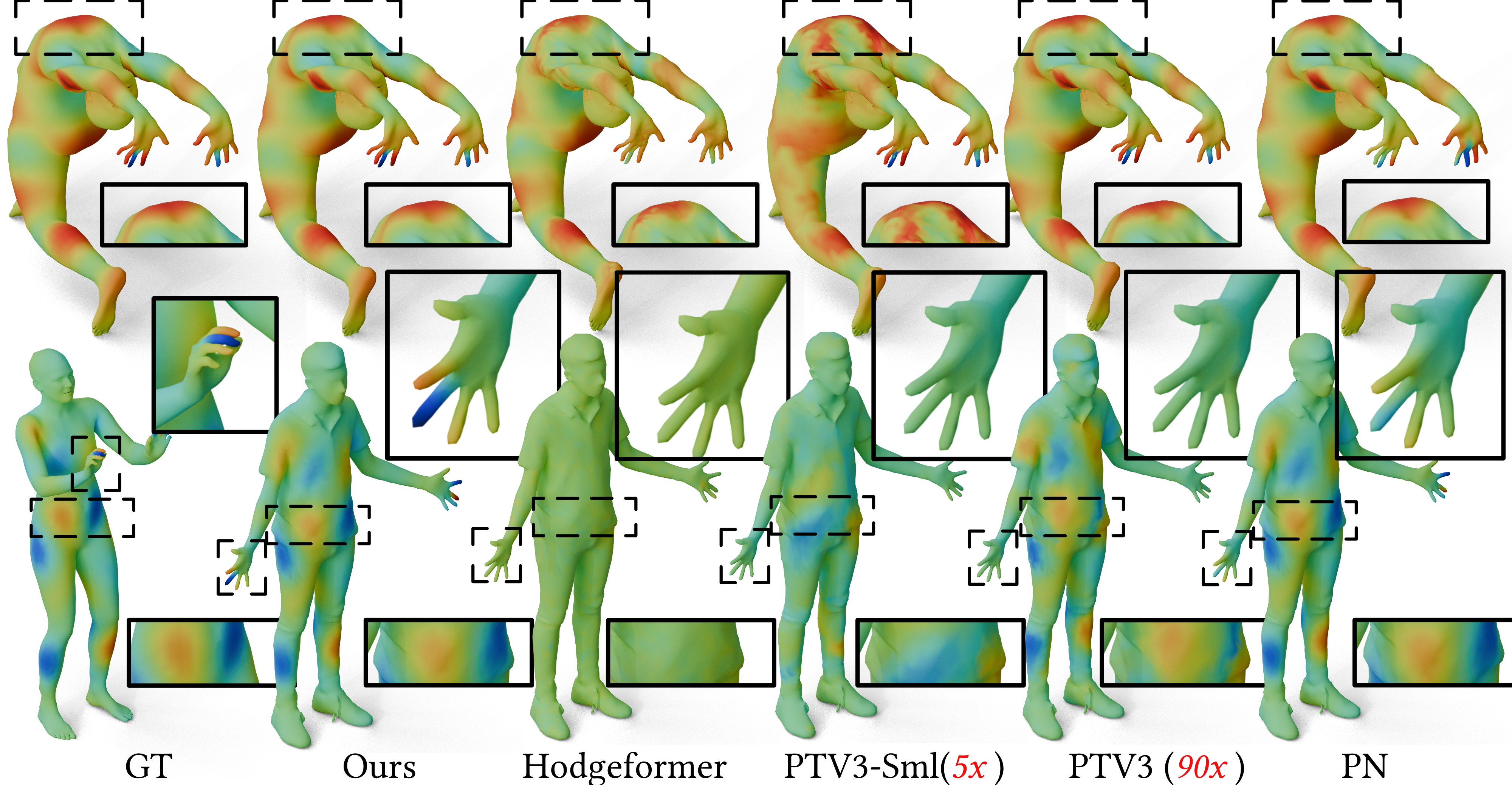}
    \caption{
  \textbf{Comparison on  high frequency signal prediction (Section~\ref{sec:pose-encoding}).} Our method accurately predicts high-frequency functions.  PoissonNet~\cite{poissonet} makes incorrect predictions on intricate areas such as the fingers, while hodgeformer~\cite{hodgeformer} and PTV3~\cite{PTV3} (both with {5x} and \textbf{90x} more parameters) are not triangulation-agnostic and thus make incorrect predictions on differently-triangulated meshes (bottom row). }
    \label{fig:6_qual_pos}
\end{figure}

%% file: figures/qual_position_encoding_flex.tex
\begin{figure}[t]
\centering
\includegraphics[width=1\linewidth]{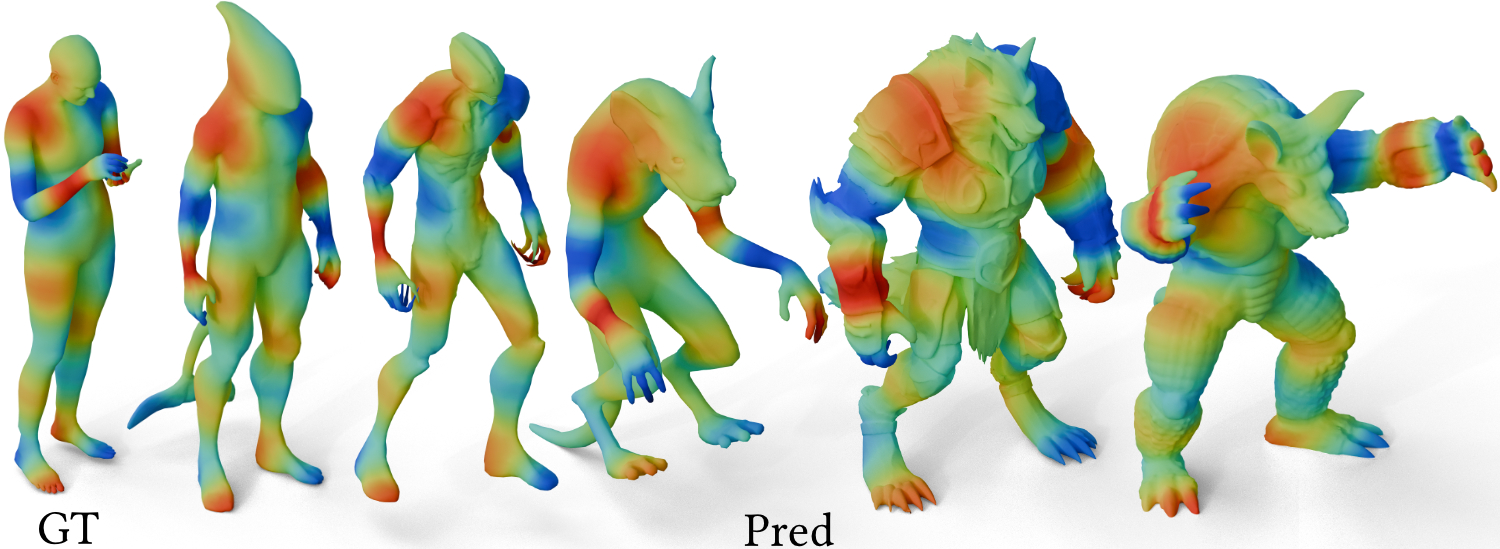}
   \caption{
   \textbf{Prediction of high frequency intrinsic functions (Section~\ref{sec:pose-encoding} for out-of-distribution shapes. }
Our model, trained solely on SMPL~\cite{SMPL-X_pavlakos_2019} meshes, makes highly-accurate predictions of high Frequency signals on shapes that are significantly different (prediction of same signal shown on all shapes).}
    \label{fig:qual_pos_enc}
\end{figure}

%% file: figures/single_source_ca_results.tex
\begin{figure}[t]
\centering
    \includegraphics[width=\linewidth]{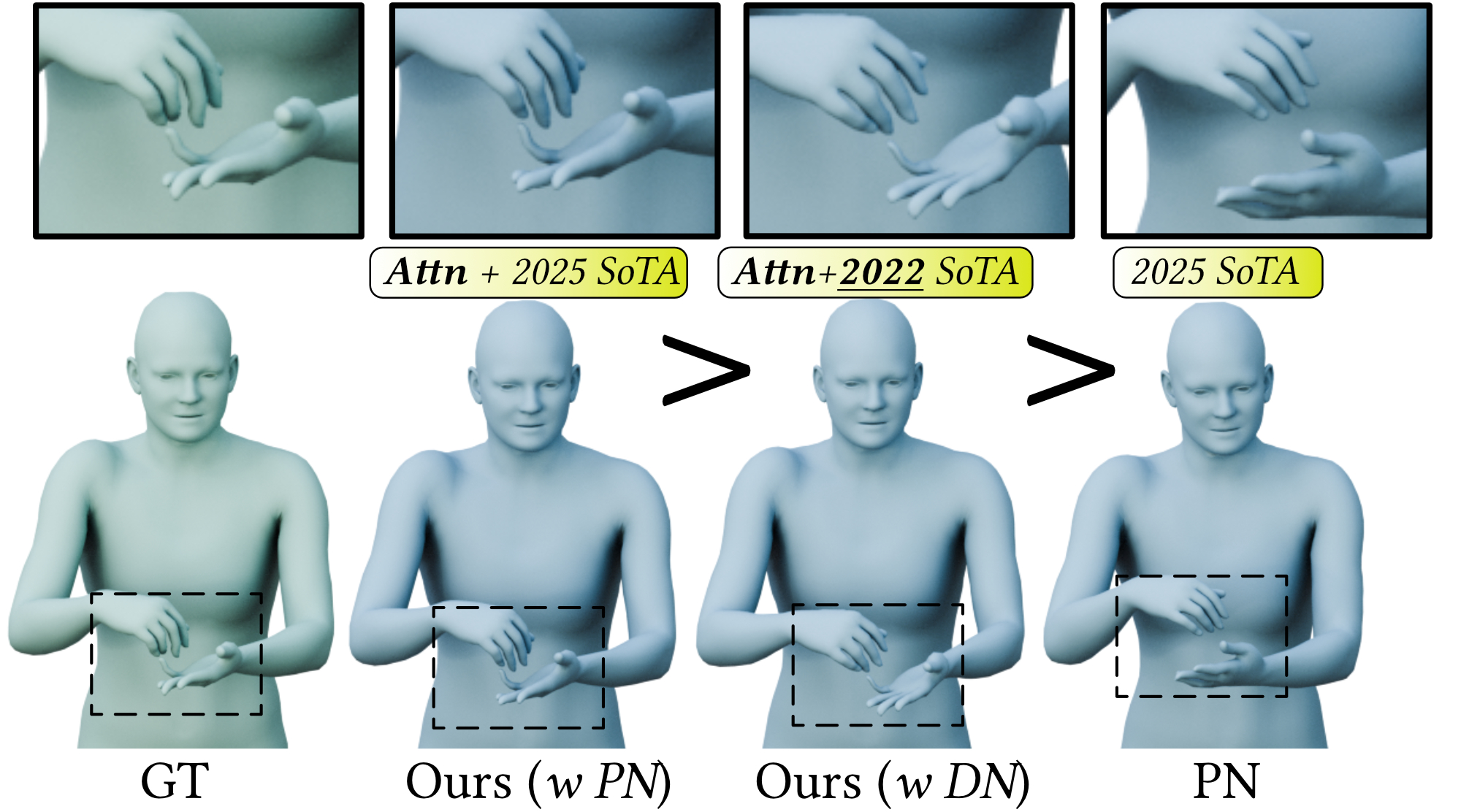}
    \caption{
    \textbf{Improvement in performance for new and old methods via our attention layers, shown on the deformation task.}
    Incorporating our simple layers leads to an immediate improvement for an older method from \textit{several years ago}, DiffusionNet~\cite{diffusionnet}  (\textit{DN}), which outperforms the \textit{current} state of the art, PoissonNet~\cite{poissonet} (\textit{PN}); note articulation of fingers. Incorporating our layers with PoissonNet itself leads to an even stronger result.
    }

    \label{fig:5_single_source_qual}

\end{figure}

%% file: figures/random_source_qual.tex
\begin{figure*}[t]
\centering
\includegraphics[width=1\textwidth]{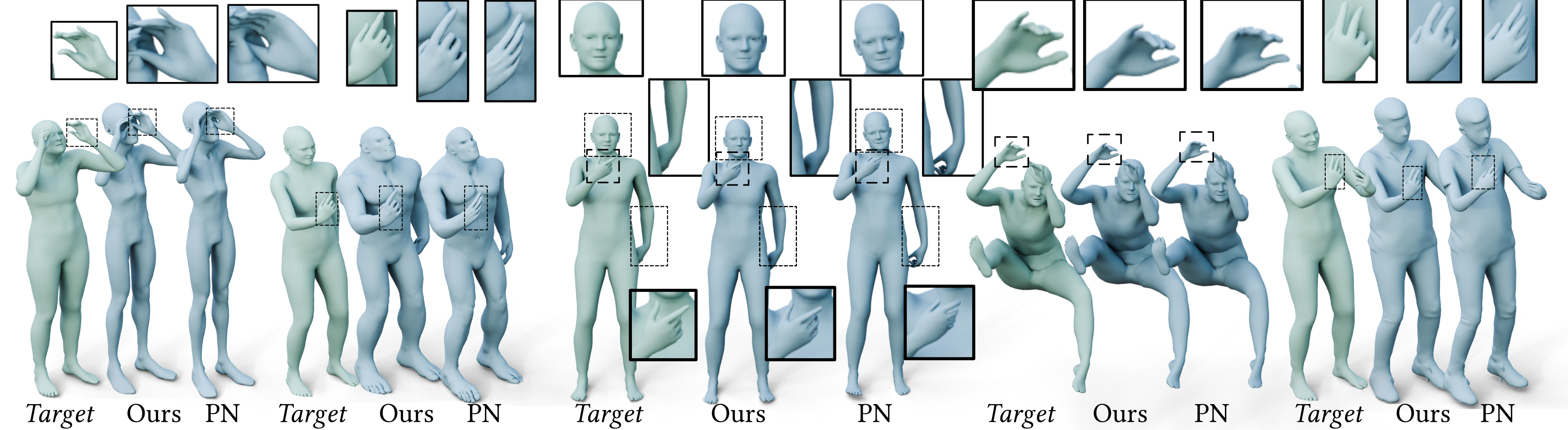}
    \caption{
   \textbf{Articulations of Arbitrary Source Meshes.} Our method can articulate arbitrary humanoids from diverse sources (artist-generated, 3D Scanned, AI-generated), without the use of rigs or training on these shapes. Each triplet shows the ground truth target position in green, our prediction, and the prediction of PoissonNet~\cite{poissonet}, which lacks the ability to predict correct fine-grained deformations (e.g., hands).
   }

\label{fig:application_articulations}
    \Description{TBD}
\end{figure*}

%% file: tables/quan_tria_agnotism.tex
\begin{table}[t]
\centering
\small
\begin{tabular}{lrrrr}
\toprule
Method & Original & Coarsened & Subdivided & Var.\ Density \\
\midrule
Naive Attention & 33.69 & 24.24 & 32.58 & 22.40 \\
\textbf{Ours} & \textbf{35.63} & \textbf{34.92} & \textbf{33.94} & \textbf{33.47} \\
\bottomrule
\end{tabular}
\caption{
\textbf{Robustness to triangulation (Section~\ref{sec:disc-agnostic-validation}).}
a network trained with our proposed attention layers on a single discretization is stable across changes to the original discretization, whereas naive attention breaks down (see difference in var density, coarsened)
}
\label{tab:robustness}
\end{table}

%% file: tables/quan_position_encoding.tex

\begin{table}[t]
\centering
\small
\begin{tabular}{lccr}
\toprule
Method & MSE ($\times 10^{-3}$) $\downarrow$ & PSNR $\uparrow$ & \#Params \\
\midrule
PoissonNet~\cite{poissonet} (= Ours w/o Attention) & 1.19 & 29.2 & 700k \\
Hodgeformer~\cite{hodgeformer} & 9.00 & 20.0 & 741k \\
PTv3~\cite{PTV3} $\times 5$ params & 1.97 & 27.1 & 4,000k \\
PTv3~\cite{PTV3} $\times 90$ params & 0.40 & 33.8 & 46,000k \\
Ours & \textbf{0.29} & \textbf{35.2} & \textbf{696k} \\
\bottomrule
\end{tabular}
\caption{
\textbf{Comparison to state-of-the-art methods on intrinsic signal prediction (Section~\ref{sec:pose-encoding}).} our attention layers attain the highest accuracy while utilizing the lowest number of parameters.
}
\label{tab:1_pos_encoding}
\end{table}

%% file: tables/quan_single_source.tex
\begin{table}[t]
\centering
\small
\begin{tabular}{lrrr}
\toprule
Method & CD $\downarrow$ & HD $\downarrow$ & MSE $\downarrow$ \\
\midrule
DiffusionNet~\cite{diffusionnet}          & 40.8 & 7.6 & 10.3 \\
PoissonNet~\cite{poissonet}            & 33.2 & 6.0 & 8.9 \\
Ours (w DiffusionNet)   & 30.0 & 5.9 & 7.0 \\
\textbf{Ours (w PoissonNet)} & \textbf{19.0} & \textbf{4.4} & \textbf{4.1} \\
\bottomrule
\end{tabular}
\caption{
\textbf{Usefulness of the attention layers with different backbones for the Deformation task (Section~\ref{sec:shape_deform}).}
Combining our attention layers with DiffusionNet would have been feasible in 2022, leading to an improvement in results compared to 2025's state-of-the-art (PoissonNet). Adding attention to the latter improves results further.
}
\label{tab:single_source}
\end{table}

%% file: figures/random_source_qual_wo_pe.tex
\begin{figure*}[!tbp]
\centering
\includegraphics[width=\textwidth]{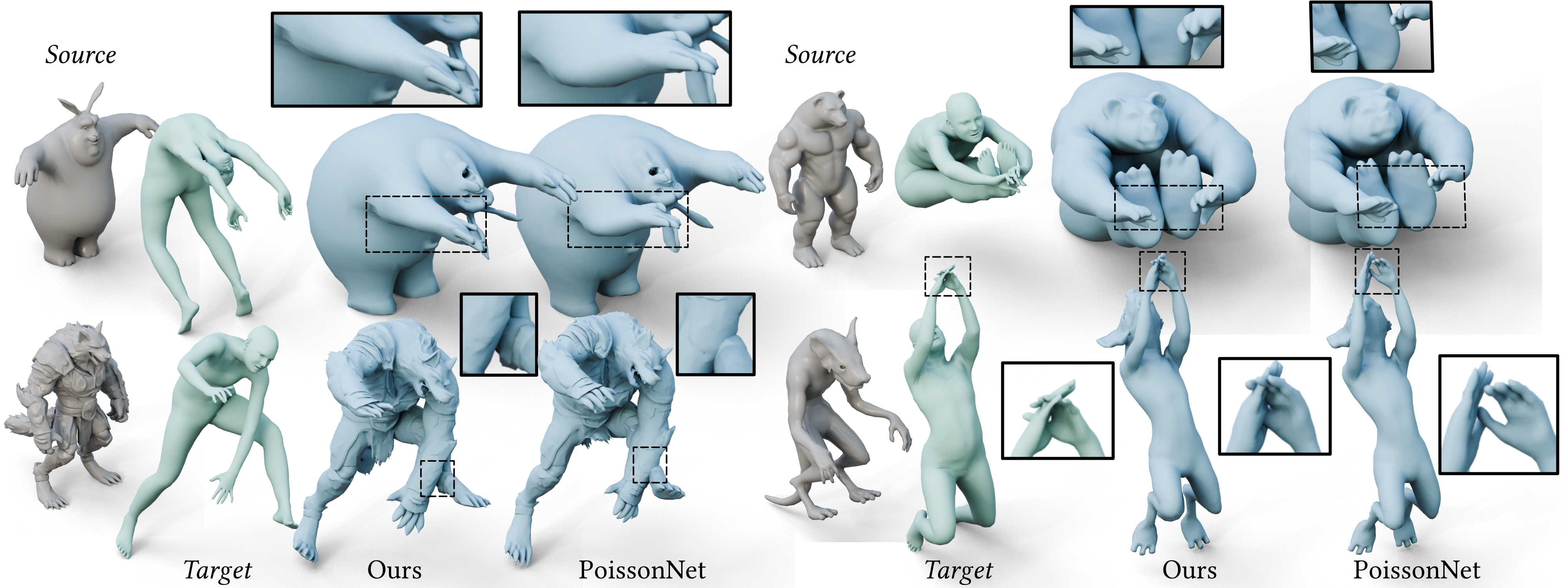}
    \caption{
   \textbf{Generalization of our deformation network to out-of-distribution shapes.} Our method, trained solely on human SMPL~\cite{SMPL-X_pavlakos_2019} meshes, generalizes to other types of humanoids, and exhibits more accurate deformations than PoissonNet~\cite{poissonet}. Target pose shown in green next to predicted deformations.
    }
    
\label{fig:qual_comparison_random_source}
    \Description{TBD}
\end{figure*}

%% file: tables/quan_random_source.tex
\begin{table}[t]
\centering
\small
\begin{tabular}{lrcr}
\toprule
Method & CD $\downarrow$ & HD $\downarrow$ & MSE $\downarrow$ \\
\midrule
Ours w/o PE \& CA (= PoissonNet) & 12.2 & 68.0 & 16.5 \\
Ours w/o PE & 9.0 & 47.0 & 10.7 \\

Ours w/o CA \& SA  & 11.6 & 64.0 & 16.4 \\
Ours w/o CA  & 10.0 & 56.0 & 14.0 \\
Ours & \textbf{7.2} & \textbf{38.7 }& \textbf{8.0} \\
\bottomrule
\end{tabular}
\caption{
\textbf{Ablation study on the mesh deformation task (Section~\ref{sec:shape_deform}).}
We evaluate the importance of each element of our deformation pipeline - self-attention (SA), cross-attention (CA), and Positional Encoding (PE). Both attention mechanisms significantly improve performance. 
}
\label{tab:3_random_source}
\end{table}

%% file: tables/quan_correspondances.tex




\begin{table}[t]
\centering
\small
\begin{tabular}{lccccc}
\toprule
Method & FAUST & SCAPE & SH19 & D-Intra & D-Inter \\
\midrule

Ini + Zoom (Lap.) & 3.8 & 7.5 & 13.1 & 1.8 & 16.5 \\
Ini + Zoom (Res.) & 3.2 & 5.7 & 12.4 & \textbf{1.6} & 13.4 \\
Smooth Shells~\cite{smoothshells} & 2.5 & 4.7 & 12.2 & / & / \\

\midrule
3D-CODED~\cite{3dcoded}& 7.5 & 17.2 & 13.4 & 45.0 & 61.4 \\
Simpl. Fmaps~\cite{simplified_fmaps} & 1.7 & 2.3 & 3.4 & 2.0 & {8.9} \\
SNK ~\cite{attaiki2023shape}& 1.8 & 4.7 & 5.8 & 2.0 & 9.0 \\

\midrule
DiffuMatch~\cite{diffumatch}& 1.9 & 4.4 & 3.9 & 1.8 & 8.6 \\
Ours & \textbf{1.24} & \textbf{1.89} & \textbf{2.83} & 2.97 & \textbf{3.31} \\

\bottomrule
\end{tabular}
\caption{
\textbf{Comparison to other methods on the correspondence benchmark from DiffuMatch~\cite{diffumatch}.} Our naïve approach attains the best results on 4 out of 5 datasets.}
\label{tab:quan_corresponance}
\end{table}

%% file: figures/qual_flex_correspondances.tex
\begin{figure*}[t]
\centering
\includegraphics[width=1\textwidth]{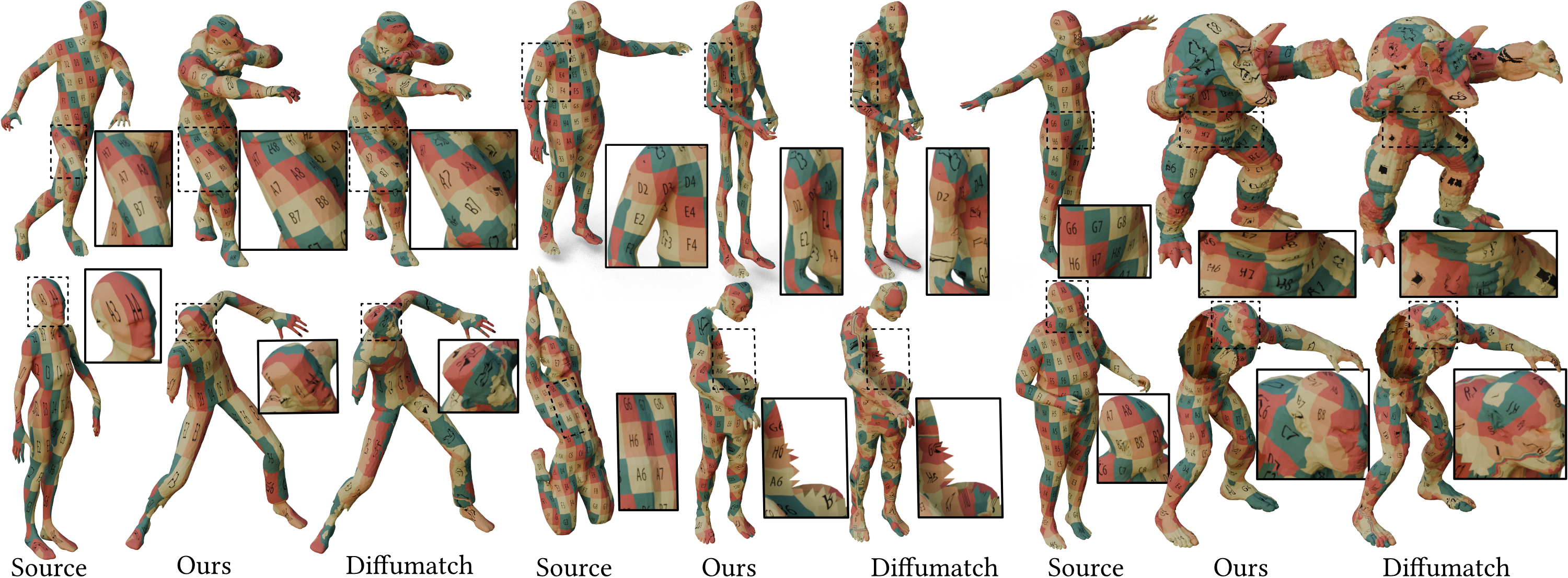}
    \caption{
   \textbf{Computed correspondences between pairs of meshes.} Our method produces dense, smooth, and continuous correspondences (visualized by transferring a texture from one mesh to the other using the correspondence map). Diffumatch~\cite{diffumatch} exhibits discontinuities and erroneous matchings.
}
    \label{fig:qual_corr}
\end{figure*}

%% file: figures/hks_results.tex
\begin{figure*}[t]
\centering
    \includegraphics[width=1\textwidth]{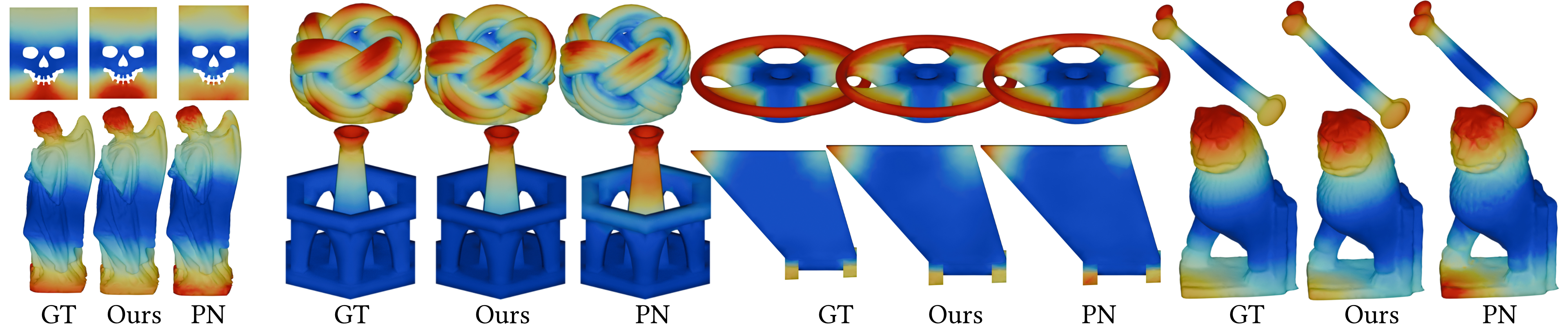}
    \caption{
    \textbf{Predicted heat kernel signature~\cite{heatKernelSignature} features over general meshes.} Our method reliably predicts intrinsic features on a dataset of general meshes, matching the ground truth significantly better than PoissonNet (PN)~\cite{poissonet}.
    }

    \label{fig:hks_results}
\end{figure*}

%% file: figures/application_position_encoding.tex
\begin{figure}[t]
\centering
    \includegraphics[width=\linewidth]{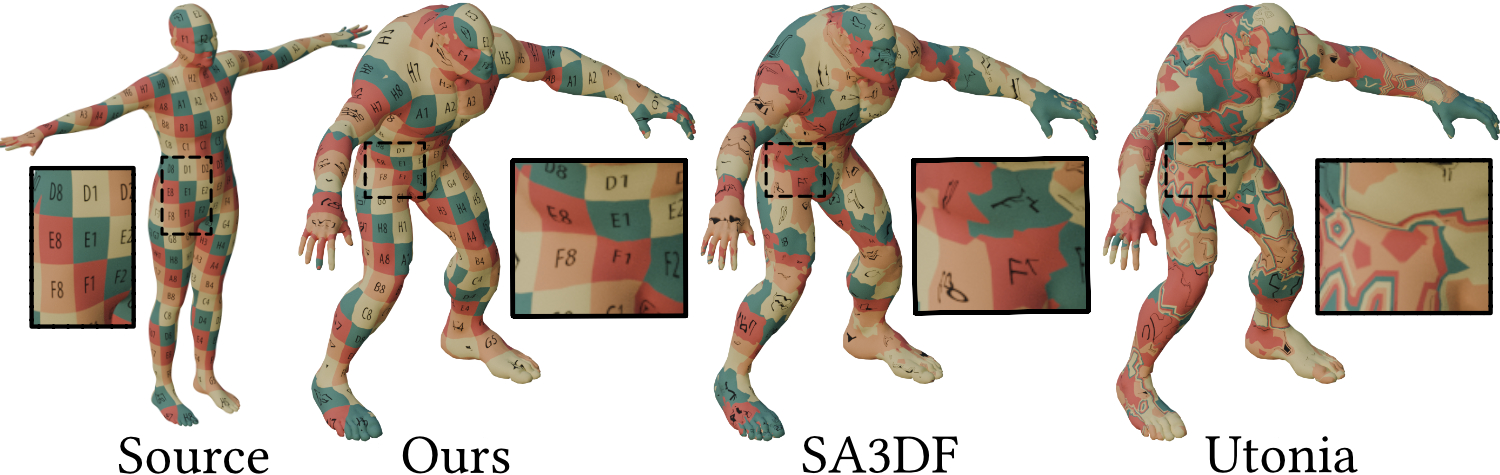}
   \caption{
   \textbf{Comparison against foundational feature extractors on the correspondence task (Section~\ref{sec:correspondance}).} The features from our method provide significantly better correspondences than both a foundational point cloud feature extractors~\cite{utonia}, as well as a 2D feature extract on meshes~\cite{sa3df}. Our method is the only one that produces smooth, continuous matching.
    }
\label{fig:6_application_position_encoding}
\end{figure}

%% file: tables/quan_hks.tex
\begin{table}[t]
\centering
\small
\begin{tabular}{lcc}
\toprule
Method & MSE ($\times 10^{-5}$) $\downarrow$ & PSNR $\uparrow$  \\
\midrule
PoissonNet & 2.8  & 15 \\
Ours & 1.8  & 17 \\
\bottomrule
\end{tabular}
\caption{
\textbf{Comparison on prediction of Heat Kernel  Signatures~\cite{heatKernelSignature}:} Our attention layers result in a significant $\times2$ improvement in accuracy.  
}
\label{tab:hks_results_quan}
\end{table}

%% file: tables/runtime.tex




\begin{table}[t]
\centering
\small
\setlength{\tabcolsep}{4pt} 
\begin{tabular}{lrrrr}
\toprule
Method & \#Params & 2k faces & 10k faces & 300k faces \\
\midrule

\multicolumn{5}{c}{\textit{High Frequency Intrinsic Signals}} \\
PoissonNet & 700k & 3.0 (25) & 6.3 (148) & 101 (2700) \\
Ours (PoissonNet + SA) & 696k & 3.5 (25) & 10.1 (147) & 1039 (2600) \\

\midrule
\multicolumn{5}{c}{\textit{Deformation}} \\
PoissonNet & 1.3M & 5.9 (28) & 12.0 (151) & 184 (2700) \\
Ours (PoissonNet + CA) & 1.5M & 6.0 (28) & 12.8 (151) & 200 (2700) \\

\bottomrule
\end{tabular}
\caption{\textbf{Runtime and memory usage of the attention layers.}
We apply our attention to PoissonNet~\cite{poissonet} and evaluate the increase in runtime and memory usage.
For each entry, we state the inference time in milliseconds along with peak GPU memory in megabytes in brackets, over several mesh resolutions.  Cross-attention (CA) comes at a minimal extra cost, however self-attention (SA) leads to a significant slowdown. Neither leads to higher memory consumption.
}
\label{tab:runtime}
\end{table}

%% file: sections/conclusion.tex
\section{Conclusion}
Our intrinsic, triangulation-agnostic attention layers provide a significant improvement to mesh-based learning tasks. We expect their use to lead to new applications, similar to the gains that attention layers have brought to other domains. For one, we are excited to use the attention layers as a backbone for mesh representation learning or in a large generative model that operates natively over meshes. Similarly, we note that we only explored conditioning on SMPL pose parameters for deformation, however cross attention opens up possibilities for prediction on meshes conditioned on other modalities such as vision and language.  We further note that the shape correspondence experiment has been set up in a very naive, straightforward manner; we believe that revisiting it with a more tailor-made architecture and loss can significantly increase its performance and scope, and we mark it as important future work.

Our work does hold several limitations. For one,  the use of PoissonNet~\cite{poissonet} as a backbone stipulates  we must operate over manifold meshes of a single connected component, which limits our applicability. In order to alleviate this, we aim in the future to extend the attention mechanism to propagate information between different connected components. A second issue is the scalability of our method, which still cannot operate on very large meshes in practical runtimes, both due to PoissonNet's PDE solves, as well as our self-attention mechanism. However we note we successfully ran experiments on most common mesh learning baselines without difficulty. Lastly, we note that our method for computing correspondence has no guarantees regarding injectivity or continuity of the resulting mesh-to-mesh map, and we mark it as important future work.

%% file: sections/7_figures.tex

    
    

    
    
    



%% file: sections/supp.tex
\input{sections/subparts/cross_attention_proof}

\input{sections/subparts/self_attention_proof}

\input{sections/subparts/other_supp_stuff}

%% file: sections/subparts/cross_attention_proof.tex
\section{Error Asymptotics For The Discretized Cross Attention}
\label{sec:cross_attn_proof}

We prove that the discretization of cross-attention, Section~\ref{sec:cross_attn} leads to an $\bigo{\varepsilon}$ error. 

Consider the softmax-weighted aggregation
\[
    f(\mathbf{p}) \;=\; \frac{\displaystyle\sum_j e^{\langle \mathbf{p},\mathbf{k}_j\rangle}\mathbf{v}_j}{\displaystyle\sum_k e^{\langle \mathbf{p},\mathbf{k}_k\rangle}},
\]
with keys $\{\mathbf{k}_j\}$ and values $\{\mathbf{v}_j\}$ fixed and bounded. Set $\varepsilon = \|\mathbf{E}\|$ and $\kappa = \max_j\|\mathbf{k}_j\|$.

\begin{proposition}
$\bigl\|f(\mathbf{p}+\mathbf{E}) - f(\mathbf{p})\bigr\| = O(\varepsilon)$ as $\varepsilon \to 0$.
\end{proposition}

\begin{proof}
Write $\alpha_j = e^{\langle\mathbf{p},\mathbf{k}_j\rangle}$, $\beta_j = e^{\langle\mathbf{E},\mathbf{k}_j\rangle}$, $A = \sum_k\alpha_k$, and $B = \sum_k\alpha_k\beta_k$, so that $e^{\langle\mathbf{p}+\mathbf{E},\,\mathbf{k}_j\rangle} = \alpha_j\beta_j$. Combining over a common denominator,
\[
    f(\mathbf{p}+\mathbf{E}) - f(\mathbf{p})
    \;=\; \frac{\displaystyle A\sum_j\alpha_j\beta_j\mathbf{v}_j \;-\; B\sum_j\alpha_j\mathbf{v}_j}{AB}.
\]
Replacing $A$ and $B$ in the numerator by their definitions and exchanging the order of summation,
\begin{align*}
    A\sum_j\alpha_j\beta_j\mathbf{v}_j - B&\sum_j\alpha_j\mathbf{v}_j
    \;=\; \sum_{j,k}\alpha_k\alpha_j\beta_j\mathbf{v}_j - \sum_{j,k}\alpha_k\beta_k\alpha_j\mathbf{v}_j
    \;=\;\\ &\sum_{j,k}\alpha_j\alpha_k(\beta_j-\beta_k)\,\mathbf{v}_j,
    \end{align*}

which gives
\begin{equation}\label{eq:exact}
    f(\mathbf{p}+\mathbf{E}) - f(\mathbf{p})
    \;=\; \frac{\displaystyle\sum_{j,k}\alpha_j\alpha_k(\beta_j-\beta_k)\,\mathbf{v}_j}{AB}.
\end{equation}
By Cauchy--Schwarz, $|\langle\mathbf{E},\mathbf{k}_j\rangle| \leq \varepsilon\kappa = O(\varepsilon)$ for all $j$. Since $e^t = 1 + O(t)$ as $t \to 0$, applying this with $t = \langle\mathbf{E},\mathbf{k}_j\rangle$ gives
\[
    \beta_j \;=\; e^{\langle\mathbf{E},\mathbf{k}_j\rangle} \;=\; 1 + O(\varepsilon),
\]
hence $\beta_j - \beta_k = O(\varepsilon)$ for every pair $(j,k)$, and $B = \sum_k\alpha_k\beta_k = A\bigl(1 + O(\varepsilon)\bigr)$. Since $A = \sum_k e^{\langle\mathbf{p},\mathbf{k}_k\rangle}$ is a finite sum of strictly positive terms depending only on the fixed quantities $\mathbf{p}$ and $\{\mathbf{k}_k\}$, it is a positive constant with respect to $\varepsilon$, so $1/A^2 = O(1)$. Taking norms of both sides of~\eqref{eq:exact}, factoring out the positive scalar $1/AB$, and applying the triangle inequality to the sum in the numerator,
\begin{align*}
    \bigl\|f(\mathbf{p}+\mathbf{E}) - &f(\mathbf{p})\bigr\|
    \;=\; \frac{\displaystyle\left\|\sum_{j,k}\alpha_j\alpha_k(\beta_j-\beta_k)\,\mathbf{v}_j\right\|}{AB}
    \;\leq\;\\&\frac{\displaystyle\sum_{j,k}\alpha_j\alpha_k\,|\beta_j-\beta_k|\,\|\mathbf{v}_j\|}{AB}.
\end{align*}
Substituting $|\beta_j - \beta_k| = O(\varepsilon)$ and $AB = A^2(1+O(\varepsilon))$,
\[
    \bigl\|f(\mathbf{p}+\mathbf{E}) - f(\mathbf{p})\bigr\|
    \;\leq\;
    \frac{O(\varepsilon)\displaystyle\sum_{j,k}\alpha_j\alpha_k\|\mathbf{v}_j\|}
         {A^2\bigl(1+O(\varepsilon)\bigr)}
    \;=\; O(\varepsilon),
\]
where $\sum_{j,k}\alpha_j\alpha_k\|\mathbf{v}_j\|$ is a finite sum of constants, $1/A^2$ is positive and fixed, and $(1+O(\varepsilon))^{-1} \to 1$, so all remaining factors are $O(1)$ in $\varepsilon$.
\end{proof}

%% file: sections/subparts/self_attention_proof.tex
\section{Error Asymptotics For The Discretized Self Attention}
\label{sec:self_attn_proof}
Let $h>0$ be the maximum edge length of the mesh. We prove the the discretiation of self-attention, Section~\ref{sec:self_attn}, leads to an $O(h^2 + \varepsilon)$ error.

Assume $\cquery,\ckey, \cvalue \in C^2(\Omega)$. Further, per the assumption of triangulation-agnosticism of the input discrete quantities, assume the discrete vertex values  $\dquery_i, \dkey_j, \dvalue_j$ approximate the
continuous fields $\cquery(\spoint_i), \ckey(\spoint_j), \cvalue(\spoint_j)$ at the vertices
with error $\varepsilon > 0$. 
\begin{proposition}
$\bigl\|\cfeature(\spoint_i) - \dfeature_i\bigr\| = O(h^2 + \varepsilon)$, as $h\to 0$.
\end{proposition}
\begin{proof}
 
Let $\dquery_i = \cquery(\spoint_i) + \delta q_i$
and $\dkey_j = \ckey(\spoint_j) + \delta k_j$ with
$\|\delta q_i\|, \|\delta k_j\| \leq \varepsilon$. Hence,
\[
  \langle \dquery_i, \dkey_j \rangle
  = \langle \cquery(\spoint_i), \ckey(\spoint_j) \rangle
  + \underbrace{\langle \cquery(\spoint_i), \delta k_j \rangle
              + \langle \delta q_i, \ckey(\spoint_j) \rangle}_{O(\varepsilon)}
  + \underbrace{\langle \delta q_i, \delta k_j \rangle}_{\bigo{\varepsilon^2}},
\]
where the $\bigo{\varepsilon}$ term dominates since the continuous $\cquery\parr{\cdot}$ and $\ckey\parr{\cdot}$
are bounded. Considering the computed discrete weights at the vertices
$\dweight_{ij} = e^{\langle \dquery_i,\dkey_j\rangle}$ compared with the continuous weights from the continuous functions, evaluated at the vertices $w_{ij} = e^{\langle \cquery(\spoint_i),\, \ckey(\spoint_j)\rangle}$, they satisfy
\[
  \dweight_{ij} = e^{\langle \cquery(\spoint_i), \ckey(\spoint_j) \rangle
  +\bigo{\varepsilon}} = w_{ij}\,e^{O(\varepsilon)} = w_{ij}(1 + O(\varepsilon))
  = w_{ij} + O(\varepsilon),
\]
 We now evaluate the accuracy of our approximation:
\[
  \cfeature(\spoint_i)
  = \frac{\displaystyle\int_\Omega w(\spoint_i,t)\,\cvalue(t)\,\mathrm{d}t}
         {\displaystyle\int_\Omega w(\spoint_i,t)\,\mathrm{d}t},
  \qquad w(\spoint_i,t) = e^{\langle \cquery(\spoint_i),\,\ckey(t)\rangle},
\]
by
\[
  \dfeature_i
  = \frac{\displaystyle\sum_j \mass_{jj}\,\dweight_{ij}\,\dvalue_j}
         {\displaystyle\sum_k \mass_{kk}\,\dweight_{ik}}.
\]

Using the  diagonal entries $\mass_{jj}$ to define a vertex quadrature
rule, standard FEM theory guarantees that for any $g\in C^2(\Omega)$,
\begin{equation}\label{eq:quad}
  \left|\int_\Omega g(t)\,\mathrm{d}t - \sum_j \mass_{jj}\,g(\spoint_j)\right| = O(h^2).
\end{equation}

Since $\dweight_{ij} = w_{ij} + O(\varepsilon)$ and $\dvalue_j =
\cvalue(\spoint_j) + O(\varepsilon)$, we show the numerator and denominator of
$\dfeature_i$ are $O(h^2 + \varepsilon)$-accurate estimates of those
of $\cfeature(\spoint_i)$. Specifically, the integrands
$w(\spoint_i,\cdot)\,\cvalue(\cdot)$ and $w(\spoint_i,\cdot)$ are both in $C^2(\Omega)$
(since $\ckey, \cvalue \in C^2$ and
$w = e^{\langle \cquery(\spoint_i), \ckey(\cdot)\rangle}$),
so~\eqref{eq:quad} together with the $O(\varepsilon)$ perturbation in
the weights and values at the vertices gives
\begin{align*}
   \int_\Omega w(\spoint_i,t)\,\cvalue(t)\,\mathrm{d}t  &=\sum_j \mass_{jj}\,w(\spoint_i,\spoint_j)\,\cvalue(\spoint_j)+ O(h^2)\\&=\sum_j \mass_{jj}\,\dweight_{ij}\,\dvalue_j+ O(h^2 + \varepsilon) 
  ,\\
  \int_\Omega w(\spoint_i,t)\,\mathrm{d}t  &= \sum_j \mass_{jj}\, w(\spoint_i,\spoint_j)
  + O(h^2). \\&=\sum_j \mass_{jj}\,\dweight_{ij}
  + O(h^2 + \varepsilon).
\end{align*}

Since $w(\spoint_i,t) = e^{\langle \cquery(\spoint_i),\ckey(t)\rangle} > 0$
for all $t\in\Omega$, the continuous denominator satisfies
\[
  D_{\min} \triangleq \int_\Omega w(\spoint_i,t)\,\mathrm{d}t
  \;>\; 0.
\]
For $h$ and $\varepsilon$ small enough that the $O(h^2+\varepsilon)$
denominator error is less than $D_{\min}/2$, the discrete denominator
is bounded below by $D_{\min}/2>0$. We finish the proof by noting that a ratio of two quantities each approximated
to accuracy $O(h^2+\varepsilon)$, with denominator bounded away from
zero, is itself approximated to accuracy $O(h^2 + \varepsilon)$.
\end{proof}

%% file: sections/subparts/other_supp_stuff.tex
\subsection{Architecture Details}
\label{sec:architecture}
For all PoissonNet \cite{poissonet} and DiffusionNet \cite{diffusionnet} backends in our experiments, we mimic their architecture exactly with the sole addition of our attention layers. We insert our attention layers (both cross- and self-attention) immediately after their respective blocks, utilizing a skip connection to join the block's output to the output of our layer. Our attention layers use a single attention head. For the cross-attention layer, we set the channel dimension to 128, and for self-attention, we set it to 126 (to ensure divisibility by 3). We utilize xformers \cite{xformer} as our memory-efficient backend. To compute the area-weighted softmax, we first add an epsilon of $1\mathrm{e}{-8}$ to the lumped mass matrix to prevent division by zero, and then add the log of this matrix; this is mathematically equivalent to our area-weighted softmax

\subsection{Dataset Details}
We adopt the train/test splits established by PoissonNet, which provide a diverse range of full-body poses from yoga motions. However, because these lack variations in finger articulations, we supplement the data with poses from GRAB, a dataset of human-object interactions. For the testing set, we utilize poses interacting with the following objects: "mug," "wineglass," "camera," "binoculars," "frying pan," and "toothpaste." To ensure diversity, we subsample the poses from the GRAB dataset using farthest point sampling. Together, the training dataset comprises 21K poses (7K MOYO and 14K GRAB), and the validation set comprises 4K poses (2K MOYO and 2K GRAB).

\subsection{Single Source  Experiment Details}
\subsubsection{Training Details}
\label{single_source_supp}
We utilized all the 21K poses from our dataset. For both the DiffusionNet and PoissonNet backbones our architecture utilizes 5 blocks. Similar to PoissonNet \cite{poissonet}, we utilize an NJF \cite{NJF} head to predict the deformations. All models are trained with the same losses used by PoissonNet for their deformation experiments, alongside an additional per-vertex loss on the hand vertices, weighted by $\lambda = 1$, to encourage the model to learn better articulations. To find the hand vertices, we use a segmentation of the SMPL mesh obtained from \cite{poissonet}. All models are trained for 200k iterations with a batch size of 16 on a single NVIDIA L40S GPU; we optimize the network using Adam with a constant learning rate of 0.005 and report the results of the best-performing model. 
\subsection{Random Source  Experiment Details}
\subsubsection{Training Details}
We sample 64K random source-target pairs from our dataset and add a shape variation with a standard deviation of 5 to the SMPL shape parameters. As in the previous experiment, all models are trained with an NJF head, and with the same losses used by PoissonNet alongside the hand vertex loss weighted by $\lambda = 1$. All models are trained for 200k iterations with a batch size of 16 on a single NVIDIA L40S GPU; we optimize the network using Adam with a constant learning rate of 0.005 and report the results of the best-performing model. 

\subsection{High Frequency Intrinsic Feature Experiment Details}
\subsubsection{Training Details }
We utilized all the 21K poses from our dataset and added a shape variation with a standard deviation of 5 to the SMPL shape parameters. Our model architecture utilizes 3 PoissonNet blocks. As in the previous experiments, all models are trained for 200k iterations with a batch size of 16 on a single NVIDIA L40S GPU, we optimize the network using Adam with a constant learning rate of 0.005, and report the results of the best performing model. For our loss function, we use a simple per-vertex Mean Squared Error (MSE) loss on the first 64 ground truth eigenfunctions for the T-pose SMPL.

\subsubsection{Baseline Details}
Here, we present details of baselines used in our experiments. 

\paragraph{PTV3:} For PTV3 \cite{PTV3}, we train their original encoder with $46\mathrm{M}$ parameters for 200K iterations and report the best results. We experiment with both flat learning rates and the cyclic learning rate schedule adopted from their ScanNet experiments, utilizing the official Pointcept repository \cite{pointcept}. We train variants of their method: one utilizing only vertex positions, and another utilizing both surface normals and vertex positions. We find that normals perform better as they provide the network with essential shape information. As demonstrated in our experiments (Sec. \ref{sec:pose-encoding}), their model overfits to the training discretization and completely fails outside of it. To curb this overfitting, we modified their model to reduce the parameter size to $5\mathrm{M}$, but found that it fails to converge even on the original discretization.

\paragraph{Hodgeformers:} We use the publicly available implementation provided by the authors and train it on our dataset. Since the method does not natively support batch sizes greater than 1, we use gradient accumulation to achieve equivalent effective batch sizes of 4 and 8. We experiment with both the learning rate settings reported for their segmentation experiments and the same learning rate used for our method. However, in both cases, the model fails to converge to satisfactory results.

\subsection{Correspondence Experiments Details}
\subsubsection{Training Details}
For this experiment, we train our model from Sec. \ref{sec:pose-encoding} with rotation augmentations uniformly sampled from $\mathrm{SO}(3)$ space and random face deletion. Specifically, we uniformly drop faces with a ratio between 0.0 and 0.2. If this random removal splits the mesh into multiple disconnected components, the entire sample is discarded. This augmentation strategy makes the network more robust to noisy data and applicable to partial meshes. We do not train on any other correspondence data or any datasets from \cite{diffumatch}.

\subsubsection{Comparison Details:}
We use the pretrained models provided by \cite{diffumatch} and are able to reproduce nearly all of their reported results. For evaluation, we use their official evaluation script and replace the correspondence maps with those computed by our method described in Sec ~\ref{sec:correspondance}.Furthermore, we observe a substantial difference in inference speed. DiffuMatch relies on Score Distillation Sampling (SDS), which is computationally expensive and takes approximately 2 minutes per sample on their evaluation dataset. In contrast, our method is highly efficient, requiring only 0.2 seconds per sample

\subsubsection{Partial Matching:}
To evaluate the robustness of our network on partial meshes, we randomly remove faces from the test dataset used in Section~\ref{sec:pose-encoding}, where the removal ratio is uniformly sampled between $0.0$ and $0.3$. We then measure the average geodesic error on both the original full meshes and the resulting partial meshes. We observe only a small degradation in performance, with the error changing from $0.5579$ on full meshes to $0.7533$ on partial meshes.



%% file: references.bib
@inproceedings{flashattention,
  title={Flash{A}ttention: Fast and Memory-Efficient Exact Attention with {IO}-Awareness},
  author={Dao, Tri and Fu, Daniel Y. and Ermon, Stefano and Rudra, Atri and R{\'e}, Christopher},
  booktitle={Advances in Neural Information Processing Systems (NeurIPS)},
  year={2022}
}

@misc{utonia,
      title={Utonia: Toward One Encoder for All Point Clouds}, 
      author={Yujia Zhang and Xiaoyang Wu and Yunhan Yang and Xianzhe Fan and Han Li and Yuechen Zhang and Zehao Huang and Naiyan Wang and Hengshuang Zhao},
      year={2026},
      eprint={2603.03283},
      archivePrefix={arXiv},
      primaryClass={cs.CV},
      url={https://arxiv.org/abs/2603.03283}, 
}

@inproceedings{sa3df,
    author = {Uzolas, Lukas and Eisemann, Elmar and Kellnhofer, Petr},
    title = {Surface-Aware Distilled 3D Semantic Features},
    year = {2025},
    isbn = {9798400721373},
    publisher = {Association for Computing Machinery},
    address = {New York, NY, USA},
    url = {https://doi.org/10.1145/3757377.3763974},
    doi = {10.1145/3757377.3763974},
    booktitle = {Proceedings of the SIGGRAPH Asia 2025 Conference Papers},
    articleno = {3},
    numpages = {12},
    location = {
    },
    series = {SA Conference Papers '25}
}

@inproceedings{moyo,
    title = {{3D} Human Pose Estimation via Intuitive Physics},
    author = {Tripathi, Shashank and M{\"u}ller, Lea and Huang, Chun-Hao P. and Taheri Omid and Black, Michael J. and Tzionas, Dimitrios},
    booktitle = {Conference on Computer Vision and Pattern Recognition ({CVPR})},
    pages = {4713--4725},
    year = {2023},
    url = {https://ipman.is.tue.mpg.de}
}

@inproceedings{GRAB:2020,
  title = {{GRAB}: A Dataset of Whole-Body Human Grasping of Objects},
  author = {Taheri, Omid and Ghorbani, Nima and Black, Michael J. and Tzionas, Dimitrios},
  booktitle = {European Conference on Computer Vision (ECCV)},
  year = {2020},
  url = {https://grab.is.tue.mpg.de}
}

@inproceedings{attention_original,
author = {Vaswani, Ashish and Shazeer, Noam and Parmar, Niki and Uszkoreit, Jakob and Jones, Llion and Gomez, Aidan N. and Kaiser, \L{}ukasz and Polosukhin, Illia},
title = {Attention is all you need},
year = {2017},
isbn = {9781510860964},
publisher = {Curran Associates Inc.},
address = {Red Hook, NY, USA},
booktitle = {Proceedings of the 31st International Conference on Neural Information Processing Systems},
pages = {6000–6010},
numpages = {11},
location = {Long Beach, California, USA},
series = {NIPS'17}
}

@inproceedings{dino,
  title={Emerging Properties in Self-Supervised Vision Transformers},
  author={Caron, Mathilde and Touvron, Hugo and Misra, Ishan and J\'egou, Herv\'e  and Mairal, Julien and Bojanowski, Piotr and Joulin, Armand},
  booktitle={Proceedings of the International Conference on Computer Vision (ICCV)},
  year={2021}
  }

@article{verma2021audio,
  title={Audio transformers},
  author={Verma, Prateek and Berger, Jonathan},
  journal={arXiv preprint arXiv:2105.00335},
  year={2021}
}

@inproceedings{PTV3,
    title={Point Transformer V3: Simpler, Faster, Stronger},
    author={Wu, Xiaoyang and Jiang, Li and Wang, Peng-Shuai and Liu, Zhijian and Liu, Xihui and Qiao, Yu and Ouyang, Wanli and He, Tong and Zhao, Hengshuang},
    booktitle={CVPR},
    year={2024}
}

@article{diffusionnet,
  title={DiffusionNet: Discretization Agnostic Learning on Surfaces},
  author={Nicholas Sharp and Souhaib Attaiki and Keenan Crane and Maks Ovsjanikov},
  journal={ACM Transactions on Graphics (TOG)},
  year={2020},
  volume={41},
  pages={1 - 16},
  url={https://api.semanticscholar.org/CorpusID:233880804}
}

@article{poissonet,
author = {Maesumi, Arman and Makadia, Tanish and Groueix, Thibault and Kim, Vladimir and Ritchie, Daniel and Aigerman, Noam},
title = {PoissonNet: A Local-Global Approach for Learning on Surfaces},
year = {2025},
issue_date = {December 2025},
publisher = {Association for Computing Machinery},
address = {New York, NY, USA},
volume = {44},
number = {6},
issn = {0730-0301},
url = {https://doi.org/10.1145/3763298},
doi = {10.1145/3763298},
journal = {ACM Trans. Graph.},
month = dec,
articleno = {175},
numpages = {16}
}

@article{hodgeformer,
  author    = {Nousias, Akis and Nousias, Stavros},
  title     = {HodgeFormer: Transformers for Learnable Operators on Triangular Meshes through Data-Driven Hodge Matrices},
  journal   = {arXiv preprint},
  year      = {2025}
}

@inproceedings{galerkin,
  author        = {Shuhao Cao},
  title         = {Choose a Transformer: {F}ourier or {G}alerkin},
  booktitle     = {Advances in Neural Information Processing Systems (NeurIPS 2021)},
  volume        = {34},
  year          = {2021},
  eprint        = {arXiv: 2105.14995},
  primaryclass  = {cs.CL},
  url={https://openreview.net/forum?id=ssohLcmn4-r},
}

@article{diffumatch,
  title={DiffuMatch: Category-Agnostic Spectral Diffusion Priors for Robust Non-rigid Shape Matching},
  author={Pierson, Emery and Li, Lei and Dai, Angela and Ovsjanikov, Maks},
  journal={arXiv preprint arXiv:2507.23715},
  year={2025}
}

@inproceedings{gpt,
  title={Improving Language Understanding by Generative Pre-Training},
  author={Alec Radford and Karthik Narasimhan},
  year={2018},
  url={https://api.semanticscholar.org/CorpusID:49313245}
}

@misc{dinov3,
  title={{DINOv3}},
  author={Sim{\'e}oni, Oriane and Vo, Huy V. and Seitzer, Maximilian and Baldassarre, Federico and Oquab, Maxime and Jose, Cijo and Khalidov, Vasil and Szafraniec, Marc and Yi, Seungeun and Ramamonjisoa, Micha{\"e}l and Massa, Francisco and Haziza, Daniel and Wehrstedt, Luca and Wang, Jianyuan and Darcet, Timoth{\'e}e and Moutakanni, Th{\'e}o and Sentana, Leonel and Roberts, Claire and Vedaldi, Andrea and Tolan, Jamie and Brandt, John and Couprie, Camille and Mairal, Julien and J{\'e}gou, Herv{\'e} and Labatut, Patrick and Bojanowski, Piotr},
  year={2025},
  eprint={2508.10104},
  archivePrefix={arXiv},
  primaryClass={cs.CV},
  url={https://arxiv.org/abs/2508.10104},
}

@article{vit,
  title={An Image is Worth 16x16 Words: Transformers for Image Recognition at Scale},
  author={Dosovitskiy, Alexey and Beyer, Lucas and Kolesnikov, Alexander and Weissenborn, Dirk and Zhai, Xiaohua and Unterthiner, Thomas and  Dehghani, Mostafa and Minderer, Matthias and Heigold, Georg and Gelly, Sylvain and Uszkoreit, Jakob and Houlsby, Neil},
  journal={ICLR},
  year={2021}
}

@inproceedings{clip,
  title={Learning Transferable Visual Models From Natural Language Supervision},
  author={Alec Radford and Jong Wook Kim and Chris Hallacy and Aditya Ramesh and Gabriel Goh and Sandhini Agarwal and Girish Sastry and Amanda Askell and Pamela Mishkin and Jack Clark and Gretchen Krueger and Ilya Sutskever},
  booktitle={International Conference on Machine Learning},
  year={2021},
  url={https://api.semanticscholar.org/CorpusID:231591445}
}

@article{LDM,
  title={High-Resolution Image Synthesis with Latent Diffusion Models},
  author={Robin Rombach and A. Blattmann and Dominik Lorenz and Patrick Esser and Bj{\"o}rn Ommer},
  journal={2022 IEEE/CVF Conference on Computer Vision and Pattern Recognition (CVPR)},
  year={2021},
  pages={10674-10685},
  url={https://api.semanticscholar.org/CorpusID:245335280}
}

@article{multimodal2,
  title={ViCA: Efficient Multimodal LLMs with Vision-Only Cross-Attention},
  author={Wenjie Liu and Hao Wu and Xin Qiu and Yingqi Fan and Yihang Zhang and Anhao Zhao and Yunpu Ma and Xiaoyu Shen},
  journal={ArXiv},
  year={2026},
  volume={abs/2602.07574},
  url={https://api.semanticscholar.org/CorpusID:285452725}
}

@unknown{xformer,
author = {Rabe, Markus and Staats, Charles},
year = {2021},
month = {12},
pages = {},
title = {Self-attention Does Not Need $O(n^2)$ Memory},
doi = {10.48550/arXiv.2112.05682}
}

@inproceedings{mamba,
  title={Transformers are {SSM}s: Generalized Models and Efficient Algorithms Through Structured State Space Duality},
  author={Dao, Tri and Gu, Albert},
  booktitle={International Conference on Machine Learning (ICML)},
  year={2024}
}

@inproceedings{LinearAttention,
  title={Transformers are RNNs: Fast Autoregressive Transformers with Linear Attention},
  author={Angelos Katharopoulos and Apoorv Vyas and Nikolaos Pappas and Franccois Fleuret},
  booktitle={International Conference on Machine Learning},
  year={2020},
  url={https://api.semanticscholar.org/CorpusID:220250819}
}

@inproceedings{Qwen,
  title={Qwen3 Technical Report},
  author={An Yang and Anfeng Li and Baosong Yang and Beichen Zhang and Binyuan Hui and Bo Zheng and Bowen Yu and Chang Gao and Chengen Huang and Chenxu Lv and Chujie Zheng and Dayiheng Liu and Fan Zhou and Fei Huang and Feng Hu and Hao Ge and Haoran Wei and Huan Lin and Jialong Tang and Jian Yang and Jianhong Tu and Jianwei Zhang and Jianxin Yang and Jiaxin Yang and Jingren Zhou and Jingren Zhou and Junyan Lin and Kai Dang and Keqin Bao and Ke‐Pei Yang and Le Yu and Li-Chun Deng and Mei Li and Min Xue and Mingze Li and Pei Zhang and Peng Wang and Qin Zhu and Rui Men and Ruize Gao and Shi-Qiang Liu and Shuang Luo and Tianhao Li and Tianyi Tang and Wenbiao Yin and Xingzhang Ren and Xinyu Wang and Xinyu Zhang and Xuancheng Ren and Yang Fan and Yang Su and Yi-Chao Zhang and Yinger Zhang and Yu Wan and Yuqiong Liu and Zekun Wang and Zeyu Cui and Zhenru Zhang and Zhipeng Zhou and Zihan Qiu},
  year={2025},
  booktitle={},
  url={https://api.semanticscholar.org/CorpusID:278602855}
}

@inproceedings{jamba,
title={Jamba: Hybrid Transformer-Mamba Language Models},
author={Barak Lenz and Opher Lieber and Alan Arazi and Amir Bergman and Avshalom Manevich and Barak Peleg and Ben Aviram and Chen Almagor and Clara Fridman and Dan Padnos and Daniel Gissin and Daniel Jannai and Dor Muhlgay and Dor Zimberg and Edden M. Gerber and Elad Dolev and Eran Krakovsky and Erez Safahi and Erez Schwartz and Gal Cohen and Gal Shachaf and Haim Rozenblum and Hofit Bata and Ido Blass and Inbal Magar and Itay Dalmedigos and Jhonathan Osin and Julie Fadlon and Maria Rozman and Matan Danos and Michael Gokhman and Mor Zusman and Naama Gidron and Nir Ratner and Noam Gat and Noam Rozen and Oded Fried and Ohad Leshno and Omer Antverg and Omri Abend and Or Dagan and Orit Cohavi and Raz Alon and Ro'i Belson and Roi Cohen and Rom Gilad and Roman Glozman and Shahar Lev and Shai Shalev-Shwartz and Shaked Haim Meirom and Tal Delbari and Tal Ness and Tomer Asida and Tom Ben Gal and Tom Braude and Uriya Pumerantz and Josh Cohen and Yonatan Belinkov and Yuval Globerson and Yuval Peleg Levy and Yoav Shoham},
booktitle={The Thirteenth International Conference on Learning Representations},
year={2025},
url={https://openreview.net/forum?id=JFPaD7lpBD}
}

@inproceedings{PTV2,
  title     = {Point transformer V2: Grouped Vector Attention and Partition-based Pooling},
  author    = {Wu, Xiaoyang and Lao, Yixing and Jiang, Li and Liu, Xihui and Zhao, Hengshuang},
  booktitle = {NeurIPS},
  year      = {2022}
}

@inproceedings{litept,
    title={{LitePT: Lighter Yet Stronger Point Transformer}},
    author={Yue, Yuanwen and Robert, Damien and Wang, Jianyuan and Hong, Sunghwan and Wegner, Jan Dirk and Rupprecht, Christian and Schindler, Konrad},
    booktitle={IEEE/CVF Conference on Computer Vision and Pattern Recognition (CVPR)},
    year={2026}
}

@article{PT,
  title={Point Transformer},
  author={Hengshuang Zhao and Li Jiang and Jiaya Jia and Philip H. S. Torr and Vladlen Koltun},
  journal={2021 IEEE/CVF International Conference on Computer Vision (ICCV)},
  year={2020},
  pages={16239-16248},
  url={https://api.semanticscholar.org/CorpusID:229220595}
}

@article{meshformer,
  title={MeshFormer: High-Quality Mesh Generation with 3D-Guided Reconstruction Model},
  author={Minghua Liu and Chong Zeng and Xinyue Wei and Ruoxi Shi and Linghao Chen and Chao Xu and Mengqi Zhang and Zhaoning Wang and Xiaoshuai Zhang and Isabella Liu and Hongzhi Wu and Hao Su},
  journal={arXiv preprint arXiv:2408.10198},
  year={2024}
}

@article{VoxFormer,
  title={VoxFormer: Sparse Voxel Transformer for Camera-Based 3D Semantic Scene Completion},
  author={Yiming Li and Zhiding Yu and Christopher Bongsoo Choy and Chaowei Xiao and Jos{\'e} Manuel {\'A}lvarez and Sanja Fidler and Chen Feng and Anima Anandkumar},
  journal={2023 IEEE/CVF Conference on Computer Vision and Pattern Recognition (CVPR)},
  year={2023},
  pages={9087-9098},
  url={https://api.semanticscholar.org/CorpusID:257102923}
}

@InProceedings{METRO2,
    title={Cross-Attention of Disentangled Modalities for 3D Human Mesh Recovery with Transformers},
    author={Junhyeong Cho and Kim Youwang and Tae-Hyun Oh},
    booktitle={European Conference on Computer Vision (ECCV)},
    year={2022}
}

@inproceedings{metro1,
author = {Lin, Kevin and Wang, Lijuan and Liu, Zicheng},
title = {End-to-End Human Pose and Mesh Reconstruction with Transformers},
booktitle = {CVPR},
year = {2021},
}

@inproceedings{meshmae,
  title={MeshMAE: Masked Autoencoders for 3D Mesh Data Analysis},
  author={Liang, Yaqian and Zhao, Shanshan and Yu, Baosheng and Zhang, Jing and He, Fazhi},
  booktitle={European Conference on Computer Vision},
  year={2022},
}

@inproceedings{meshmamba,
  title={Mesh Mamba: A Unified State Space Model for Saliency Prediction in Non-Textured and Textured Meshes},
  author={Zhang, Kaiwei and Zhu, Dandan and Min, Xiongkuo and Zhai, Guangtao},
  booktitle={Proceedings of the Computer Vision and Pattern Recognition Conference},
  pages={16219--16228},
  year={2025}
}

@article{meshgraphtransformer,
author = {Vecchio, Giuseppe and Prezzavento, Luca and Pino, Carmelo and Rundo, Francesco and Palazzo, Simone and Spampinato, Concetto},
title = { MeT: A graph transformer for semantic segmentation of 3D meshes},
year = {2023},
issue_date = {Oct 2023},
publisher = {Elsevier Science Inc.},
address = {USA},
volume = {235},
number = {C},
issn = {1077-3142},
url = {https://doi.org/10.1016/j.cviu.2023.103773},
doi = {10.1016/j.cviu.2023.103773},
journal = {Comput. Vis. Image Underst.},
month = oct,
numpages = {9}
}

@inproceedings{laplacianmeshtransformer,
author = {Li, Xiao-Juan and Yang, Jie and Zhang, Fang-Lue},
title = {Laplacian Mesh Transformer: Dual Attention and Topology Aware Network for 3D Mesh Classification and Segmentation},
year = {2022},
isbn = {978-3-031-19817-5},
publisher = {Springer-Verlag},
address = {Berlin, Heidelberg},
url = {https://doi.org/10.1007/978-3-031-19818-2_31},
doi = {10.1007/978-3-031-19818-2_31},
booktitle = {Computer Vision – ECCV 2022: 17th European Conference, Tel Aviv, Israel, October 23–27, 2022, Proceedings, Part XXIX},
pages = {541–560},
numpages = {20},
location = {Tel Aviv, Israel}
}

@article{recipe,
  title={A Recipe for Geometry-Aware 3D Mesh Transformers},
  author={Mohammad Farazi and Yalin Wang},
  journal={2025 IEEE/CVF Winter Conference on Applications of Computer Vision (WACV)},
  year={2024},
  pages={3290-3300},
  url={https://api.semanticscholar.org/CorpusID:273798014}
}

@article{neuraloperator,
   author    = {Nikola Kovachki and
                  Zongyi Li and
                  Burigede Liu and
                  Kamyar Azizzadenesheli and
                  Kaushik Bhattacharya and
                  Andrew Stuart and
                  Anima Anandkumar},
   title     = {Neural Operator: Learning Maps Between Function Spaces with Applications to PDEs},
   journal   = {JMLR},
   volume    = {24},
   number    = {1},
   articleno = {89},
   numpages  = {97},
   year      = {2023},
}

@inproceedings{Transolve,
  title={Transolver: A Fast Transformer Solver for PDEs on General Geometries},
  author={Haixu Wu and Huakun Luo and Haowen Wang and Jianmin Wang and Mingsheng Long},
  booktitle={International Conference on Machine Learning},
  year={2024},
  url={https://api.semanticscholar.org/CorpusID:267411758}
}

@inproceedings{sun2024tuttenet,
  title={TutteNet: Injective 3D Deformations by Composition of 2D Mesh Deformations},
  author={Sun, Bo and Groueix, Thibault and Song, Chen and Huang, Qixing and Aigerman, Noam},
  booktitle={Proceedings of the IEEE/CVF Conference on Computer Vision and Pattern Recognition},
  pages={21378--21389},
  year={2024}
}

@article{meshcnn_hanocka_2019,
author = {Hanocka, Rana and Hertz, Amir and Fish, Noa and Giryes, Raja and Fleishman, Shachar and Cohen-Or, Daniel},
title = {MeshCNN: a network with an edge},
year = {2019},
issue_date = {August 2019},
publisher = {Association for Computing Machinery},
address = {New York, NY, USA},
volume = {38},
number = {4},
issn = {0730-0301},
url = {https://doi.org/10.1145/3306346.3322959},
doi = {10.1145/3306346.3322959},
journal = {ACM Trans. Graph.},
month = jul,
articleno = {90},
numpages = {12}
}

@article{hodgenet_smirnov_2021,
author = {Smirnov, Dmitriy and Solomon, Justin},
title = {HodgeNet: learning spectral geometry on triangle meshes},
year = {2021},
issue_date = {August 2021},
publisher = {Association for Computing Machinery},
address = {New York, NY, USA},
volume = {40},
number = {4},
issn = {0730-0301},
url = {https://doi.org/10.1145/3450626.3459797},
doi = {10.1145/3450626.3459797},
journal = {ACM Trans. Graph.},
month = jul,
articleno = {166},
numpages = {11}
}

@inproceedings{attaiki2021dpfm,
  title={Dpfm: Deep partial functional maps},
  author={Attaiki, Souhaib and Pai, Gautam and Ovsjanikov, Maks},
  booktitle={2021 International Conference on 3D Vision (3DV)},
  pages={175--185},
  year={2021},
  organization={IEEE}
}

@inproceedings{donati2020deep,
  title={Deep geometric functional maps: Robust feature learning for shape correspondence},
  author={Donati, Nicolas and Sharma, Abhishek and Ovsjanikov, Maks},
  booktitle={Proceedings of the IEEE/CVF Conference on Computer Vision and Pattern Recognition},
  pages={8592--8601},
  year={2020}
}

@inproceedings{roufosse2019unsupervised,
  title={Unsupervised deep learning for structured shape matching},
  author={Roufosse, Jean-Michel and Sharma, Abhishek and Ovsjanikov, Maks},
  booktitle={Proceedings of the IEEE/CVF International Conference on Computer Vision},
  pages={1617--1627},
  year={2019}
}

@inproceedings{halimi2019unsupervised,
  title={Unsupervised learning of dense shape correspondence},
  author={Halimi, Oshri and Litany, Or and Rodola, Emanuele and Bronstein, Alex M and Kimmel, Ron},
  booktitle={Proceedings of the IEEE/CVF Conference on Computer Vision and Pattern Recognition},
  pages={4370--4379},
  year={2019}
}

@inproceedings{geodesicConvo_masci_2015,
  title={Geodesic convolutional neural networks on riemannian manifolds},
  author={Masci, Jonathan and Boscaini, Davide and Bronstein, Michael and Vandergheynst, Pierre},
  booktitle={Proceedings of the IEEE international conference on computer vision workshops},
  pages={37--45},
  year={2015}
}

@inproceedings{he2020curvanet,
  title={CurvaNet: Geometric deep learning based on directional curvature for 3D shape analysis},
  author={He, Wenchong and Jiang, Zhe and Zhang, Chengming and Sainju, Arpan Man},
  booktitle={Proceedings of the 26th ACM SIGKDD International Conference on Knowledge Discovery \& Data Mining},
  pages={2214--2224},
  year={2020}
}

@article{de2020gauge,
  title={Gauge equivariant mesh cnns: Anisotropic convolutions on geometric graphs},
  author={De Haan, Pim and Weiler, Maurice and Cohen, Taco and Welling, Max},
  journal={arXiv preprint arXiv:2003.05425},
  year={2020}
}

@InProceedings{Litany_2017_ICCV,
author = {Litany, Or and Remez, Tal and Rodola, Emanuele and Bronstein, Alex and Bronstein, Michael},
title = {Deep Functional Maps: Structured Prediction for Dense Shape Correspondence},
booktitle = {Proceedings of the IEEE International Conference on Computer Vision (ICCV)},
month = {Oct},
year = {2017}
}

@InProceedings{Yi_2017_CVPR,
author = {Yi, Li and Su, Hao and Guo, Xingwen and Guibas, Leonidas J.},
title = {SyncSpecCNN: Synchronized Spectral CNN for 3D Shape Segmentation},
booktitle = {Proceedings of the IEEE Conference on Computer Vision and Pattern Recognition (CVPR)},
month = {July},
year = {2017}
}

@article{ovsjanikov2012functional,
  title={Functional maps: a flexible representation of maps between shapes},
  author={Ovsjanikov, Maks and Ben-Chen, Mirela and Solomon, Justin and Butscher, Adrian and Guibas, Leonidas},
  journal={ACM Transactions on Graphics (ToG)},
  volume={31},
  number={4},
  pages={1--11},
  year={2012},
  publisher={ACM New York, NY, USA}
}

@inproceedings{monti2017geometric,
  title={Geometric deep learning on graphs and manifolds using mixture model cnns},
  author={Monti, Federico and Boscaini, Davide and Masci, Jonathan and Rodola, Emanuele and Svoboda, Jan and Bronstein, Michael M},
  booktitle={Proceedings of the IEEE conference on computer vision and pattern recognition},
  pages={5115--5124},
  year={2017}
}

@inproceedings{fey2018splinecnn,
  title={Splinecnn: Fast geometric deep learning with continuous b-spline kernels},
  author={Fey, Matthias and Lenssen, Jan Eric and Weichert, Frank and M{\"u}ller, Heinrich},
  booktitle={Proceedings of the IEEE conference on computer vision and pattern recognition},
  pages={869--877},
  year={2018}
}

@article{boscaini2016learning,
  title={Learning shape correspondence with anisotropic convolutional neural networks},
  author={Boscaini, Davide and Masci, Jonathan and Rodol{\`a}, Emanuele and Bronstein, Michael},
  journal={Advances in neural information processing systems},
  volume={29},
  year={2016}
}

@article{bronstein2017geometric,
  title={Geometric deep learning: going beyond euclidean data},
  author={Bronstein, Michael M and Bruna, Joan and LeCun, Yann and Szlam, Arthur and Vandergheynst, Pierre},
  journal={IEEE Signal Processing Magazine},
  volume={34},
  number={4},
  pages={18--42},
  year={2017},
  publisher={IEEE}
}

@inproceedings{dynamic_simonovsky_2017,
  title={Dynamic edge-conditioned filters in convolutional neural networks on graphs},
  author={Simonovsky, Martin and Komodakis, Nikos},
  booktitle={Proceedings of the IEEE conference on computer vision and pattern recognition},
  pages={3693--3702},
  year={2017}
}

@inproceedings{SMPL-X_pavlakos_2019,
  title = {Expressive Body Capture: {3D} Hands, Face, and Body from a Single Image},
  author = {Pavlakos, Georgios and Choutas, Vasileios and Ghorbani, Nima and Bolkart, Timo and Osman, Ahmed A. A. and Tzionas, Dimitrios and Black, Michael J.},
  booktitle = {Proceedings IEEE Conf. on Computer Vision and Pattern Recognition (CVPR)},
  pages     = {10975--10985},
  year = {2019}
}

@article{meshwalker_lahav_2020,
  title={Meshwalker: Deep mesh understanding by random walks},
  author={Lahav, Alon and Tal, Ayellet},
  journal={ACM Transactions on Graphics (TOG)},
  volume={39},
  number={6},
  pages={1--13},
  year={2020},
  publisher={ACM New York, NY, USA}
}

@article{fieldconvo_mitchel_2021,
    author    = {Mitchel, Thomas W. and Kim, Vladimir G. and Kazhdan, Michael},
    title     = {Field Convolutions for Surface CNNs},
    booktitle = {Proceedings of the IEEE/CVF International Conference on Computer Vision (ICCV)},
    month     = {October},
    year      = {2021},
    pages     = {10001-10011}
}

@article{HSN_wiersma_2020,
  title={CNNs on surfaces using rotation-equivariant features},
  author={Wiersma, Ruben and Eisemann, Elmar and Hildebrandt, Klaus},
  journal={ACM Transactions on Graphics (ToG)},
  volume={39},
  number={4},
  year={2020},
  publisher={ACM New York, NY, USA}
}

@inproceedings{cgconv_Yang_2021,
author = {Yang, Zhangsihao and Litany, Or and Birdal, Tolga and Sridhar, Srinath and Guibas, Leonidas},
year = {2021},
month = {01},
pages = {134-144},
title = {Continuous Geodesic Convolutions for Learning on 3D Shapes},
doi = {10.1109/WACV48630.2021.00018}
}

@article{mdgcnn_poulenard_2018,
author = {Poulenard, Adrien and Ovsjanikov, Maks},
title = {Multi-directional geodesic neural networks via equivariant convolution},
year = {2018},
issue_date = {December 2018},
publisher = {Association for Computing Machinery},
address = {New York, NY, USA},
volume = {37},
number = {6},
issn = {0730-0301},
url = {https://doi.org/10.1145/3272127.3275102},
doi = {10.1145/3272127.3275102},
journal = {ACM Trans. Graph.},
month = dec,
articleno = {236},
numpages = {14}
}

@article{deltaconv,
  title={Deltaconv: anisotropic operators for geometric deep learning on point clouds},
  author={Wiersma, Ruben and Nasikun, Ahmad and Eisemann, Elmar and Hildebrandt, Klaus},
  journal={ACM Transactions on Graphics (TOG)},
  volume={41},
  number={4},
  pages={1--10},
  year={2022},
  publisher={ACM New York, NY, USA}
}

@article{temporalNJF,
  title={Temporal Residual Jacobians for Rig-Free Motion Transfer},
  author={Muralikrishnan, Sanjeev and Dutt, Niladri and Chaudhuri, Siddhartha and Aigerman, Noam and Kim, Vladimir and Fisher, Matthew and Mitra, Niloy J},
  booktitle={European Conference on Computer Vision},
  pages={93--109},
  year={2024},
  organization={Springer}
}

@article{NJF,
  title={Neural jacobian fields: Learning intrinsic mappings of arbitrary meshes},
  author={Aigerman, Noam and Gupta, Kunal and Kim, Vladimir G and Chaudhuri, Siddhartha and Saito, Jun and Groueix, Thibault},
  journal={SIGGRAPH},
  year={2022}
}

@article{Gao23,
title = {TextDeformer: Geometry Manipulation using Text Guidance},
author = {William Gao and Noam Aigerman and Thibault Groueix and Vladimir G. Kim and Rana Hanocka},
year = {2023},
journal = {SIGGRAPH (Conference track)}}

@article{Kim25,
title = {MeshUp: Multi-Target Mesh Deformation via Blended Score Distillation},
author = {Hyunwoo Kim and Itai Lang and Thibault Groueix and Noam Aigerman and Vladimir G. Kim and Rana Hanocka},
year = {2025},
journal = {3DV}}

@article{Yoo24,
title = {As-Plausible-As-Possible: Plausibility-Aware Mesh Deformation Using 2D Diffusion Priors},
author = {Seungwoo Yoo and Kunho Kim and Vladimir G. Kim and Minhyuk Sung},
year = {2024},
journal = {CVPR}}

@INPROCEEDINGS{Lipman04,
  author={Lipman, Y. and Sorkine, O. and Cohen-Or, D. and Levin, D. and Rossi, C. and Seidel, H.P.},
  booktitle={Proceedings Shape Modeling Applications, 2004.}, 
  title={Differential coordinates for interactive mesh editing}, 
  year={2004},
  pages={181-190},
}

@inproceedings{Sorkine04,
author = {Sorkine, O. and Cohen-Or, D. and Lipman, Y. and Alexa, M. and R\"{o}ssl, C. and Seidel, H.-P.},
title = {Laplacian surface editing},
year = {2004},
booktitle = {Proceedings of the 2004 Eurographics/ACM SIGGRAPH Symposium on Geometry Processing},
pages = {175–184},
numpages = {10},
location = {Nice, France},
series = {SGP '04}
}

@inproceedings{Yu04,
author = {Yu, Yizhou and Zhou, Kun and Xu, Dong and Shi, Xiaohan and Bao, Hujun and Guo, Baining and Shum, Heung-Yeung},
title = {Mesh editing with poisson-based gradient field manipulation},
year = {2004},
booktitle = {ACM SIGGRAPH 2004 Papers},
pages = {644–651},
numpages = {8}}

@article{sumner2004deformation,
  title={Deformation transfer for triangle meshes},
  author={Sumner, Robert W and Popovi{\'c}, Jovan},
  journal={ACM Transactions on graphics (TOG)},
  volume={23},
  number={3},
  pages={399--405},
  year={2004},
  publisher={ACM New York, NY, USA}
}

@inproceedings{bogo2014faust,
  title = {{FAUST}: Dataset and evaluation for {3D} mesh registration},
  author = {Bogo, Federica and Romero, Javier and Loper, Matthew and Black, Michael J.},
  booktitle = {CVPR},
  year = {2014}
}

@article{gao2019sdm,
  title={SDM-NET: Deep generative network for structured deformable mesh},
  author={Gao, Lin and Yang, Jie and Wu, Tong and Yuan, Yu-Jie and Fu, Hongbo and Lai, Yu-Kun and Zhang, Hao},
  journal={ACM Transactions on Graphics (TOG)},
  volume={38},
  number={6},
  pages={1--15},
  year={2019},
  publisher={ACM New York, NY, USA}
}

@inproceedings{skinningcourse:2014,
  author = {Alec Jacobson and Zhigang Deng and Ladislav Kavan and JP Lewis},
  title = {Skinning: Real-time Shape Deformation},
  booktitle = {ACM SIGGRAPH 2014 Courses},
  year = {2014},
}

@article{Fulton:LSD:2018,
  title ={Latent-space Dynamics for Reduced Deformable Simulation},
  author = {Lawson Fulton and Vismay Modi and David Duvenaud and David I. W. Levin and Alec Jacobson},
  year = {2019},
  journal = {Computer Graphics Forum}
}

@article{jacobson2011bounded,
  title={Bounded biharmonic weights for real-time deformation.},
  author={Jacobson, Alec and Baran, Ilya and Popovic, Jovan and Sorkine, Olga},
  journal={ACM Trans. Graph.},
  volume={30},
  number={4},
  pages={78},
  year={2011},
  publisher={Citeseer}
}

@article{kavan2008geometric,
  title={Geometric skinning with approximate dual quaternion blending},
  author={Kavan, Ladislav and Collins, Steven and {\v{Z}}{\'a}ra, Ji{\v{r}}{\'\i} and O'Sullivan, Carol},
  journal={ACM Transactions on Graphics (TOG)},
  volume={27},
  number={4},
  pages={1--23},
  year={2008},
  publisher={ACM New York, NY, USA}
}

@article{lipman2008green,
  title={Green coordinates},
  author={Lipman, Yaron and Levin, David and Cohen-Or, Daniel},
  journal={ACM Transactions on Graphics (TOG)},
  volume={27},
  number={3},
  pages={1--10},
  year={2008},
  publisher={ACM New York, NY, USA}
}

@article{ju2005mean,
  title={Mean value coordinates for closed triangular meshes},
  author={Ju, Tao and Schaefer, Scott and Warren, Joe},
  journal={ACM Siggraph 2005 Papers},
  pages={561--566},
  year={2005}
}

@InProceedings{AnimSkelVolNet,
    title={Predicting Animation Skeletons for 3D Articulated Models via Volumetric Nets},
    author={Zhan Xu and Yang Zhou and Evangelos Kalogerakis and Karan Singh},
    booktitle={2019 International Conference on 3D Vision (3DV)},
    year={2019}
  }

@article{RigNet,
    title={RigNet: Neural Rigging for Articulated Characters},
    author={Zhan Xu and Yang Zhou and Evangelos Kalogerakis and Chris Landreth and Karan Singh},
    journal={ACM Trans. on Graphics},
    year={2020},
    volume={39}
  }

@inproceedings{Holden:inverse_rig:2015,
author = {Holden, Daniel and Saito, Jun and Komura, Taku},
title = {Learning an Inverse Rig Mapping for Character Animation},
year = {2015},
isbn = {9781450334969},
publisher = {Association for Computing Machinery},
address = {New York, NY, USA},
url = {https://doi.org/10.1145/2786784.2786788},
doi = {10.1145/2786784.2786788},
booktitle = {Proceedings of the 14th ACM SIGGRAPH / Eurographics Symposium on Computer Animation},
pages = {165–173},
numpages = {9},
location = {Los Angeles, California},
series = {SCA '15}
}

@article{li2021learning,
  author = {Li, Peizhuo and Aberman, Kfir and Hanocka, Rana and Liu, Libin and Sorkine-Hornung, Olga and Chen, Baoquan},
  title = {Learning Skeletal Articulations with Neural Blend Shapes},
  journal = {ACM Transactions on Graphics (TOG)},
  volume = {40},
  number = {4},
  pages = {1},
  year = {2021},
  publisher = {ACM}
}

@article{liu2025riganything,
    title={RigAnything: Template-Free Autoregressive Rigging for Diverse 3D Assets},
    author={Liu, Isabella and Xu, Zhan and Yifan, Wang and Tan, Hao and Xu, Zexiang and Wang, Xiaolong and Su, Hao and Shi, Zifan},
    journal={arXiv preprint arXiv:2502.09615},
    year={2025}
}

@inproceedings{wang2019neural,
  author = {Yifan Wang and Noam Aigerman and Vladimir G. Kim and Siddhartha Chaudhuri and Olga Sorkine-Hornung},
  title = {Neural Cages for Detail-Preserving {3D} Deformations},
  booktitle = {CVPR},
  year = {2020},
}

@INPROCEEDINGS{varol17_surreal,
  title     = {Learning from Synthetic Humans},
  author    = {Varol, G{\"u}l and Romero, Javier and Martin, Xavier and Mahmood, Naureen and Black, Michael J. and Laptev, Ivan and Schmid, Cordelia},
  booktitle = {CVPR},
  year      = {2017}
}

@inproceedings{SMAL:2017,
        title = {{3D} Menagerie: Modeling the {3D} Shape and Pose of Animals},
        author = {Zuffi, Silvia and Kanazawa, Angjoo and Jacobs, David and Black, Michael J.},
        booktitle = {IEEE Conf. on Computer Vision and Pattern Recognition (CVPR)},
        month = jul,
        year = {2017},
        month_numeric = {7}
      }

@inproceedings{STAR:2020,
      author = {Osman, Ahmed A A and Bolkart, Timo and Black, Michael J.},
      title = {{STAR}: A Sparse Trained Articulated Human Body Regressor},
      booktitle = {European Conference on Computer Vision (ECCV)},
      pages = {598--613},
      year = {2020},
      url = {https://star.is.tue.mpg.de}
}

@inproceedings{anguelov2005scape,
  author = {Anguelov, Dragomir and Srinivasan, Praveen and Koller, Daphne and Thrun, Sebastian and Rodgers, Jim and Davis, James},
  title = {{SCAPE}: Shape Completion and Animation of People},
  booktitle = {SIGGRAPH},
  year = {2005},
}

@inproceedings{Bogo:ECCV:2016,
        title = {Keep it {SMPL}: Automatic Estimation of {3D} Human Pose and Shape
        from a Single Image},
        author = {Bogo, Federica and Kanazawa, Angjoo and Lassner, Christoph and
        Gehler, Peter and Romero, Javier and Black, Michael J.},
        booktitle = {Computer Vision -- ECCV 2016},
        series = {Lecture Notes in Computer Science},
        publisher = {Springer International Publishing},
        month = oct,
        year = {2016}
        }

@inproceedings{Shen:2021,
title={High-order differentiable autoencoder for nonlinear model reduction},
author = {Siyuan Shen and Yin Yang and Tianjia Shao and He Wang and Chenfanfu Jiang and Lei Lan and Kun Zhou},
journal = {ACM Transactions on Graphics},
year = {2021}
}

@inproceedings{sun2020zernet,
  title={Zernet: Convolutional neural networks on arbitrary surfaces via zernike local tangent space estimation},
  author={Sun, Zhiyu and Rooke, Ethan and Charton, Jerome and He, Yusen and Lu, Jia and Baek, Stephen},
  booktitle={Computer Graphics Forum},
  volume={39},
  number={6},
  pages={204--216},
  year={2020},
  organization={Wiley Online Library}
}

@article{Bailey:2018:FDD,
  note = {Presented at SIGGRAPH 2018, Los Angeles},
  doi = {10.1145/3197517.3201300},
  title = {Fast and Deep Deformation Approximations},
  journal = {ACM Transactions on Graphics},
  author = {Stephen W. Bailey and Dave Otte and Paul Dilorenzo and James F. O'Brien},
  number = 4,
  month = aug,
  volume = 37,
  year = 2018,
  pages = {119:1--12},
  url = {http://graphics.berkeley.edu/papers/Bailey-FDD-2018-08/},
}

@article{bailey2020fast,
  title={Fast and deep facial deformations},
  author={Bailey, Stephen W and Omens, Dalton and Dilorenzo, Paul and O'Brien, James F},
  journal={ACM Transactions on Graphics (TOG)},
  volume={39},
  number={4},
  pages={94--1},
  year={2020},
  publisher={ACM New York, NY, USA}
}

@article{Romero:2021,
  author       = "Romero, Cristian and Casas, Dan and Perez, Jesus and Otaduy, Miguel A.",
  title        = "Learning Contact Corrections for Handle-Based Subspace Dynamics",
  journal      = "ACM Trans. on Graphics (Proc. of ACM SIGGRAPH)",
  number       = "4",
  volume       = "40",
  year         = "2021",
  url          = "http://gmrv.es/Publications/2021/RCPO21"
}

@inproceedings{Zheng:secondary_motion:2021,
  title={A Deep Emulator for Secondary Motion of 3D Characters},
  author={Zheng, Mianlun and Zhou, Yi and Ceylan, Duygu and Barbic, Jernej},
  booktitle={Proceedings of the IEEE/CVF Conference on Computer Vision and Pattern Recognition},
  pages={5932--5940},
  year={2021}
}

@inproceedings{yin2021_3DStyleNet,
    title = {3DStyleNet: Creating 3D Shapes with Geometric and Texture Style Variations}, 
    author = {Kangxue Yin and Jun Gao and Maria Shugrina and Sameh Khamis and Sanja Fidler},
              booktitle = {Proceedings of International Conference on Computer Vision (ICCV)},
              year = {2021}
}

@article{Thingi10K,
  title={Thingi10K: A Dataset of 10,000 3D-Printing Models},
  author={Zhou, Qingnan and Jacobson, Alec},
  journal={arXiv preprint arXiv:1605.04797},
  year={2016}
}

@inproceedings{heatKernelSignature,
  title={A concise and provably informative multi-scale signature based on heat diffusion},
  author={Sun, Jian and Ovsjanikov, Maks and Guibas, Leonidas},
  booktitle={Computer graphics forum},
  volume={28},
  number={5},
  pages={1383--1392},
  year={2009},
  organization={Wiley Online Library}
}

@article{Litany2017DeepFM,
  title={Deep Functional Maps: Structured Prediction for Dense Shape Correspondence},
  author={Or Litany and Tal Remez and Emanuele Rodol{\`a} and Alexander M. Bronstein and Michael M. Bronstein},
  journal={2017 IEEE International Conference on Computer Vision (ICCV)},
  year={2017},
  pages={5660-5668},
  url={https://api.semanticscholar.org/CorpusID:4215682}
}

@article{attaiki2023shape,
  title={Shape non-rigid kinematics (snk): A zero-shot method for non-rigid shape matching via unsupervised functional map regularized reconstruction},
  author={Attaiki, Souhaib and Ovsjanikov, Maks},
  journal={Advances in Neural Information Processing Systems},
  volume={36},
  pages={70012--70032},
  year={2023}
}

@article{reviewnonrigid,
    title   = {Non-Rigid 3D Shape Correspondences: From Foundations to Open Challenges and Opportunities},
    author  = {Zhuravlev, Aleksei and Bastian, Lennart and Cao, Dongliang and El Amrani, Nafie and Roetzer, Paul and Ehm, Viktoria and Marin, Riccardo and Nishizawa, Hiroki and Morishima, Shigeo and Theobalt, Christian and Navab, Nassir and Cremers, Daniel and Bernard, Florian and L{\"a}hner, Zorah and Golyanik, Vladislav},
    journal = {Computer Graphics Forum},
    year    = {2026},
    doi     = {10.1111/cgf.70397}
    }

@INPROCEEDINGS{RL01,
  author={Rusinkiewicz, S. and Levoy, M.},
  booktitle={Proceedings Third International Conference on 3-D Digital Imaging and Modeling}, 
  title={Efficient variants of the ICP algorithm}, 
  year={2001},
  volume={},
  number={},
  pages={145-152},
  doi={10.1109/IM.2001.924423}}

@INPROCEEDINGS{ZSN03,
  author={Zinsser, T. and Schmidt, J. and Niemann, H.},
  booktitle={Proceedings 2003 International Conference on Image Processing (Cat. No.03CH37429)}, 
  title={A refined ICP algorithm for robust 3-D correspondence estimation}, 
  year={2003},
  volume={2},
  number={},
  pages={II-695},
  doi={10.1109/ICIP.2003.1246775}}

@inproceedings{MAKM21,
author = {Morreale, Luca and Aigerman, Noam and Kim, Vladimir and Mitra, Niloy},
year = {2021},
month = {11},
pages = {},
title = {Neural Surface Maps},
doi = {10.1109/CVPR46437.2021.00461}
}

@article{neural_semantic_surface_maps,
author = {Morreale, Luca and Aigerman, Noam and Kim, Vladimir and Mitra, Niloy},
year = {2024},
month = {04},
pages = {},
title = {Neural Semantic Surface Maps},
volume = {43},
journal = {Computer Graphics Forum},
doi = {10.1111/cgf.15005}
}

@article{sadh04,
author = {Schreiner, John and Asirvatham, Arul and Praun, Emil and Hoppe, Hugues},
title = {Inter-surface mapping},
year = {2004},
issue_date = {August 2004},
publisher = {Association for Computing Machinery},
address = {New York, NY, USA},
volume = {23},
number = {3},
issn = {0730-0301},
url = {https://doi.org/10.1145/1015706.1015812},
doi = {10.1145/1015706.1015812},
journal = {ACM Trans. Graph.},
month = aug,
pages = {870–877},
numpages = {8}
}

@article{apl14,
author = {Aigerman, Noam and Poranne, Roi and Lipman, Yaron},
title = {Lifted bijections for low distortion surface mappings},
year = {2014},
issue_date = {July 2014},
publisher = {Association for Computing Machinery},
address = {New York, NY, USA},
volume = {33},
number = {4},
issn = {0730-0301},
url = {https://doi.org/10.1145/2601097.2601158},
doi = {10.1145/2601097.2601158},
journal = {ACM Trans. Graph.},
month = jul,
articleno = {69},
numpages = {12}
}

@article{SBCK19,
author = {Schmidt, Patrick and Born, Janis and Campen, Marcel and Kobbelt, Leif},
title = {Distortion-minimizing injective maps between surfaces},
year = {2019},
issue_date = {December 2019},
publisher = {Association for Computing Machinery},
address = {New York, NY, USA},
volume = {38},
number = {6},
issn = {0730-0301},
url = {https://doi.org/10.1145/3355089.3356519},
doi = {10.1145/3355089.3356519},
journal = {ACM Trans. Graph.},
month = nov,
articleno = {156},
numpages = {15}
}

@inproceedings{AAWEOW,
author = {Abdelreheem, Ahmed and Eldesokey, Abdelrahman and Ovsjanikov, Maks and Wonka, Peter},
title = {Zero-Shot 3D Shape Correspondence},
year = {2023},
isbn = {9798400703157},
publisher = {Association for Computing Machinery},
address = {New York, NY, USA},
url = {https://doi.org/10.1145/3610548.3618228},
doi = {10.1145/3610548.3618228},
booktitle = {SIGGRAPH Asia 2023 Conference Papers},
articleno = {59},
numpages = {11},
location = {Sydney, NSW, Australia},
series = {SA '23}
}

@article{ren2018continuous,
  title={Continuous and orientation-preserving correspondences via functional maps},
  author={Ren, Jing and Poulenard, Adrien and Wonka, Peter and Ovsjanikov, Maks},
  journal={ACM Transactions on Graphics (ToG)},
  volume={37},
  number={6},
  pages={1--16},
  year={2018},
  publisher={ACM New York, NY, USA}
}

@inproceedings{panine2022non,
  title={Non-Isometric Shape Matching via Functional Maps on Landmark-Adapted Bases},
  author={Panine, Mikhail and Kirgo, Maxime and Ovsjanikov, Maks},
  booktitle={Computer graphics forum},
  volume={41},
  number={6},
  pages={394--417},
  year={2022},
  organization={Wiley Online Library}
}

@article{zoomout,
author = {Melzi, Simone and Ren, Jing and Rodol\`{a}, Emanuele and Sharma, Abhishek and Wonka, Peter and Ovsjanikov, Maks},
title = {ZoomOut: spectral upsampling for efficient shape correspondence},
year = {2019},
issue_date = {December 2019},
publisher = {Association for Computing Machinery},
address = {New York, NY, USA},
volume = {38},
number = {6},
issn = {0730-0301},
url = {https://doi.org/10.1145/3355089.3356524},
doi = {10.1145/3355089.3356524},
journal = {ACM Trans. Graph.},
month = nov,
articleno = {155},
numpages = {14}
}

@unknown{smoothshells,
author = {Eisenberger, Marvin and Lähner, Zorah and Cremers, Daniel},
year = {2019},
month = {05},
pages = {},
title = {Smooth Shells: Multi-Scale Shape Registration with Functional Maps},
doi = {10.48550/arXiv.1905.12512}
}

@inproceedings{3dcoded,
  			title={Learning elementary structures for 3D shape generation and matching},
  			author={Deprelle, Theo and Groueix, Thibault and Fisher, Matthew and Kim, Vladimir G and Russell, Bryan C and Aubry, Mathieu},
  			booktitle={Neurips},
  			year={2019}
}

@inproceedings{simplified_fmaps,
author = {Magnet, Robin and Ovsjanikov, Maks},
year = {2024},
month = {06},
pages = {4041-4050},
title = {Memory-Scalable and Simplified Functional Map Learning},
doi = {10.1109/CVPR52733.2024.00387}
}

@misc{pointcept,
    title={Pointcept: A Codebase for Point Cloud Perception Research},
    author={Pointcept Contributors},
    howpublished = {\url{https://github.com/Pointcept/Pointcept}},
    year={2023}
}
